\documentclass[conference]{IEEEtran} 

\IEEEoverridecommandlockouts

\usepackage[utf8]{inputenc}

\usepackage{amssymb,amsmath,amsthm}
\usepackage{mathtools}
\usepackage{MnSymbol} 

\newcommand\doubleplus{+\kern-1.3ex+\kern0.8ex}

\newcommand{\vect}{\overrightarrow}

\usepackage{stmaryrd} 
\usepackage{xfrac} 

\usepackage{newunicodechar}
\let\Alpha=A
\let\Beta=B
\let\Epsilon=E
\let\Zeta=Z
\let\Eta=H
\let\Iota=I
\let\Kappa=K
\let\Mu=M
\let\Nu=N
\let\Omicron=O
\let\omicron=o
\let\Rho=P
\let\Tau=T
\let\Chi=X

\newunicodechar{Α}{\ensuremath{\Alpha}}
\newunicodechar{α}{\ensuremath{\alpha}}
\newunicodechar{Β}{\ensuremath{\Beta}}
\newunicodechar{𝔹}{\ensuremath{\mathbb{B}}}
\newunicodechar{β}{\ensuremath{\beta}}
\newunicodechar{Γ}{\ensuremath{\Gamma}}
\newunicodechar{γ}{\ensuremath{\gamma}}
\newunicodechar{Δ}{\ensuremath{\Delta}}
\newunicodechar{δ}{\ensuremath{\delta}}
\newunicodechar{Ε}{\ensuremath{\Epsilon}}
\newunicodechar{ε}{\ensuremath{\epsilon}}
\newunicodechar{ϵ}{\ensuremath{\varepsilon}}
\newunicodechar{Ζ}{\ensuremath{\Zeta}}
\newunicodechar{ζ}{\ensuremath{\zeta}}
\newunicodechar{Η}{\ensuremath{\Eta}}
\newunicodechar{η}{\ensuremath{\eta}}
\newunicodechar{Θ}{\ensuremath{\Theta}}
\newunicodechar{θ}{\ensuremath{\theta}}
\newunicodechar{ϑ}{\ensuremath{\vartheta}}
\newunicodechar{Ι}{\ensuremath{\Iota}}
\newunicodechar{ι}{\ensuremath{\iota}}
\newunicodechar{Κ}{\ensuremath{\Kappa}}
\newunicodechar{κ}{\ensuremath{\kappa}}
\newunicodechar{Λ}{\ensuremath{\Lambda}}
\newunicodechar{λ}{\ensuremath{\lambda}}
\newunicodechar{Μ}{\ensuremath{\Mu}}
\newunicodechar{μ}{\ensuremath{\mu}}
\newunicodechar{Ν}{\ensuremath{\Nu}}
\newunicodechar{ν}{\ensuremath{\nu}}
\newunicodechar{Ξ}{\ensuremath{\Xi}}
\newunicodechar{ξ}{\ensuremath{\xi}}
\newunicodechar{Ο}{\ensuremath{\Omicron}}
\newunicodechar{ο}{\ensuremath{\omicron}}
\newunicodechar{Π}{\ensuremath{\Pi}}
\newunicodechar{π}{\ensuremath{\pi}}
\newunicodechar{ϖ}{\ensuremath{\varpi}}
\newunicodechar{Ρ}{\ensuremath{\Rho}}
\newunicodechar{ρ}{\ensuremath{\rho}}
\newunicodechar{ϱ}{\ensuremath{\varrho}}
\newunicodechar{Σ}{\ensuremath{\Sigma}}
\newunicodechar{σ}{\ensuremath{\sigma}}
\newunicodechar{ς}{\ensuremath{\varsigma}}
\newunicodechar{Τ}{\ensuremath{\Tau}}
\newunicodechar{τ}{\ensuremath{\tau}}
\newunicodechar{Υ}{\ensuremath{\Upsilon}}
\newunicodechar{υ}{\ensuremath{\upsilon}}
\newunicodechar{Φ}{\ensuremath{\Phi}}
\newunicodechar{φ}{\ensuremath{\varphi}}
\newunicodechar{ϕ}{\ensuremath{\phi}}
\newunicodechar{Χ}{\ensuremath{\Chi}}
\newunicodechar{χ}{\ensuremath{\chi}}
\newunicodechar{Ψ}{\ensuremath{\Psi}}
\newunicodechar{ψ}{\ensuremath{\psi}}
\newunicodechar{Ω}{\ensuremath{\Omega}}
\newunicodechar{ω}{\ensuremath{\omega}}

\newunicodechar{ℂ}{\ensuremath{\mathbb{C}}}
\newunicodechar{ℤ}{\ensuremath{\mathbb{Z}}}
\newunicodechar{ℝ}{\ensuremath{\mathbb{R}}}
\newunicodechar{ℕ}{\ensuremath{\mathbb{N}}}
\newunicodechar{ℚ}{\ensuremath{\mathbb{Q}}}

\newunicodechar{𝒜}{\ensuremath{\mathcal{A}}}
\newunicodechar{ℬ}{\ensuremath{\mathcal{B}}}
\newunicodechar{𝒞}{\ensuremath{\mathcal{C}}}
\newunicodechar{𝒟}{\ensuremath{\mathcal{D}}}
\newunicodechar{ℰ}{\ensuremath{\mathcal{E}}}
\newunicodechar{ℱ}{\ensuremath{\mathcal{F}}}
\newunicodechar{𝒢}{\ensuremath{\mathcal{G}}}
\newunicodechar{ℋ}{\ensuremath{\mathcal{H}}}
\newunicodechar{ℑ}{\ensuremath{\mathcal{I}}}
\newunicodechar{𝒦}{\ensuremath{\mathcal{K}}}
\newunicodechar{ℒ}{\ensuremath{\mathcal{L}}}
\newunicodechar{ℳ}{\ensuremath{\mathcal{M}}}
\newunicodechar{𝒩}{\ensuremath{\mathcal{N}}}
\newunicodechar{𝒪}{\ensuremath{\mathcal{O}}}
\newunicodechar{𝒫}{\ensuremath{\mathcal{P}}}
\newunicodechar{𝒬}{\ensuremath{\mathcal{Q}}}
\newunicodechar{ℛ}{\ensuremath{\mathcal{R}}}
\newunicodechar{ℜ}{\ensuremath{\mathcal{R}}}
\newunicodechar{𝒮}{\ensuremath{\mathcal{S}}}
\newunicodechar{𝒯}{\ensuremath{\mathcal{T}}}
\newunicodechar{𝒰}{\ensuremath{\mathcal{U}}}
\newunicodechar{𝒱}{\ensuremath{\mathcal{V}}}
\newunicodechar{𝒲}{\ensuremath{\mathcal{W}}}
\newunicodechar{ℓ}{\ensuremath{\ell}}

\newunicodechar{∼}{\ensuremath{\sim}}
\newunicodechar{≈}{\ensuremath{\approx}}
\newunicodechar{≋}{\ensuremath{\triplesim}}
\newunicodechar{≅}{\ensuremath{\cong}}
\newunicodechar{≡}{\ensuremath{\equiv}}
\newunicodechar{≤}{\ensuremath{\le}}
\newunicodechar{≥}{\ensuremath{\ge}}
\newunicodechar{≲}{\ensuremath{\lesssim}}
\newunicodechar{≠}{\ensuremath{\neq}}
\newunicodechar{≔}{\ensuremath{\coloneqq}}
\newunicodechar{≜}{\ensuremath{\triangleq}}
\newunicodechar{⋆}{\ensuremath{\applysymbol}}
\newunicodechar{ƛ}{\ensuremath{\linearlambda}}
\newunicodechar{▷}{\ensuremath{\rhd}}

\newunicodechar{±}{\ensuremath{\pm}}
\newunicodechar{∓}{\ensuremath{\pm}}
\newunicodechar{×}{\ensuremath{\times}}

\newunicodechar{→}{\ensuremath{\to}}
\newunicodechar{←}{\ensuremath{\leftarrow}}
\newunicodechar{⇒}{\ensuremath{\Rightarrow}}
\newunicodechar{↦}{\ensuremath{\mapsto}}
\newunicodechar{↝}{\ensuremath{\leadsto}}
\newunicodechar{⸖}{\ensuremath{\mathrel{>!}}}
\newunicodechar{⇆}{\ensuremath{\xlrsquigarrow}}
\newunicodechar{↔}{\ensuremath{\leftrightarrow}}
\newunicodechar{⇔}{\ensuremath{\iff}}
\newunicodechar{↭}{\ensuremath{\leftrightsquigarrow}}
\newunicodechar{↣}{\ensuremath{\rightarrowtail}}
\newunicodechar{↑}{\ensuremath{\uparrow}}
\newunicodechar{⇑}{\ensuremath{\Uparrow}}
\newunicodechar{↓}{\ensuremath{\downarrow}}
\newunicodechar{⇓}{\ensuremath{\Downarrow}}
\newunicodechar{⇀}{\ensuremath{\rightharpoonup}}
\newunicodechar{↗}{\ensuremath{\nearrow}}
\newunicodechar{↛}{\ensuremath{\nrightarrow}}
\newunicodechar{↚}{\ensuremath{\nleftarrow}}

\newunicodechar{∀}{\ensuremath{\forall}}
\newunicodechar{∃}{\ensuremath{\exists}}
\newunicodechar{∨}{\ensuremath{\vee}}
\newunicodechar{∧}{\ensuremath{\wedge}}
\newunicodechar{⊢}{\ensuremath{\vdash}}
\newunicodechar{⊩}{\ensuremath{\Vdash}}
\newunicodechar{⊨}{\ensuremath{\vDash}}
\newunicodechar{⊣}{\ensuremath{\dashv}}
\newunicodechar{⊤}{\ensuremath{\top}}
\newunicodechar{⊥}{\ensuremath{\bot}}
\newunicodechar{¬}{\ensuremath{\neg}}
\newunicodechar{⊏}{\ensuremath{\sqsubset}}
\newunicodechar{⊐}{\ensuremath{\sqsupset}}

\newunicodechar{⊸}{\ensuremath{\multimap}}
\newunicodechar{⊗}{\ensuremath{\otimes}}
\newunicodechar{⊕}{\ensuremath{\oplus}}
\newunicodechar{¡}{\ensuremath{\oc}}
\newunicodechar{⨂}{\ensuremath{\bigotimes}}
\newunicodechar{⨁}{\ensuremath{\bigoplus}}

\newunicodechar{∈}{\ensuremath{\in}}
\newunicodechar{∉}{\ensuremath{\not\in}}
\newunicodechar{⊆}{\ensuremath{\subseteq}}
\newunicodechar{⊊}{\ensuremath{\subsetneq}}
\newunicodechar{⊂}{\ensuremath{\subset}}
\newunicodechar{∪}{\ensuremath{\cup}}
\newunicodechar{⋓}{\ensuremath{\Cup}}
\newunicodechar{∩}{\ensuremath{\cap}}
\newunicodechar{∅}{\ensuremath{\emptyset}}

\newunicodechar{〈}{\ensuremath{\langle}}
\newunicodechar{⟨}{\ensuremath{\langle}}
\newunicodechar{⟩}{\ensuremath{\rangle}}
\newunicodechar{⟪}{\ensuremath{\llangle}}
\newunicodechar{⟫}{\ensuremath{\rrangle}}
\newunicodechar{〉}{\ensuremath{\rangle}}
\newunicodechar{⟦}{\ensuremath{\llbracket}}
\newunicodechar{⟧}{\ensuremath{\rrbracket}}
\newunicodechar{⌈}{\ensuremath{\lceil}}
\newunicodechar{⌉}{\ensuremath{\rceil}}
\newunicodechar{⌊}{\ensuremath{\lfloor}}
\newunicodechar{⌋}{\ensuremath{\rfloor}}

\newunicodechar{₀}{\ensuremath{_0}}
\newunicodechar{₁}{\ensuremath{_1}}
\newunicodechar{₂}{\ensuremath{_2}}
\newunicodechar{₃}{\ensuremath{_3}}
\newunicodechar{₄}{\ensuremath{_4}}
\newunicodechar{₅}{\ensuremath{_5}}
\newunicodechar{₆}{\ensuremath{_6}}
\newunicodechar{₇}{\ensuremath{_7}}
\newunicodechar{₈}{\ensuremath{_8}}
\newunicodechar{₉}{\ensuremath{_9}}
\newunicodechar{ₐ}{\ensuremath{_a}}
\newunicodechar{ₑ}{\ensuremath{_e}}
\newunicodechar{ₕ}{\ensuremath{_h}}
\newunicodechar{ᵢ}{\ensuremath{_i}}
\newunicodechar{ⱼ}{\ensuremath{_j}}
\newunicodechar{ₖ}{\ensuremath{_k}}
\newunicodechar{ₗ}{\ensuremath{_l}}
\newunicodechar{ₘ}{\ensuremath{_m}}
\newunicodechar{ₙ}{\ensuremath{_n}}
\newunicodechar{ₒ}{\ensuremath{_o}}
\newunicodechar{ₚ}{\ensuremath{_p}}
\newunicodechar{ᵣ}{\ensuremath{_r}}
\newunicodechar{ₛ}{\ensuremath{_s}}
\newunicodechar{ₜ}{\ensuremath{_t}}
\newunicodechar{ᵤ}{\ensuremath{_u}}
\newunicodechar{ₓ}{\ensuremath{_x}}

\newunicodechar{⁰}{\ensuremath{^0}}
\newunicodechar{¹}{\ensuremath{^1}}
\newunicodechar{²}{\ensuremath{^2}}
\newunicodechar{³}{\ensuremath{^3}}
\newunicodechar{⁴}{\ensuremath{^4}}
\newunicodechar{⁵}{\ensuremath{^5}}
\newunicodechar{⁶}{\ensuremath{^6}}
\newunicodechar{⁷}{\ensuremath{^7}}
\newunicodechar{⁸}{\ensuremath{^8}}
\newunicodechar{⁹}{\ensuremath{^9}}
\newunicodechar{ᵃ}{\ensuremath{^a}}
\newunicodechar{ᵇ}{\ensuremath{^b}}
\newunicodechar{ᶜ}{\ensuremath{^c}}
\newunicodechar{ᵈ}{\ensuremath{^d}}
\newunicodechar{ᵉ}{\ensuremath{^e}}
\newunicodechar{ᶠ}{\ensuremath{^f}}
\newunicodechar{ᵍ}{\ensuremath{^g}}
\newunicodechar{ʰ}{\ensuremath{^h}}
\newunicodechar{ⁱ}{\ensuremath{^i}}
\newunicodechar{ʲ}{\ensuremath{^j}}
\newunicodechar{ᵏ}{\ensuremath{^k}}
\newunicodechar{ᵐ}{\ensuremath{^m}}
\newunicodechar{ⁿ}{\ensuremath{^n}}
\newunicodechar{ᵒ}{\ensuremath{^o}}
\newunicodechar{ᵖ}{\ensuremath{^p}}
\newunicodechar{ˢ}{\ensuremath{^s}}
\newunicodechar{ᵗ}{\ensuremath{^t}}
\newunicodechar{ᵘ}{\ensuremath{^u}}
\newunicodechar{ʷ}{\ensuremath{^v}}
\newunicodechar{ˣ}{\ensuremath{^x}}
\newunicodechar{ʸ}{\ensuremath{^y}}
\newunicodechar{ᶻ}{\ensuremath{^z}}

\newunicodechar{•}{\ensuremath{\bullet}}
\newunicodechar{∙}{\ensuremath{\bullet}}
\newunicodechar{·}{\ensuremath{\cdot}}
\newunicodechar{⋯}{\ensuremath{\cdots}}
\newunicodechar{…}{\ensuremath{\ldots}}
\newunicodechar{⋮}{\ensuremath{\vdots}}
\newunicodechar{∶}{\ensuremath{~:~}}
\newunicodechar{∷}{\ensuremath{~\mathrel{:\!\!\!:}~}}

\newunicodechar{∣}{\ensuremath{\mid}}
\newunicodechar{∥}{\ensuremath{ {\parallel} }}

\newunicodechar{□}{\ensuremath{\square}}
\newunicodechar{⋄}{\ensuremath{\diamond}}

\newunicodechar{∗}{\ensuremath{\ast}}

\newunicodechar{∘}{\ensuremath{\circ}}
\newunicodechar{†}{\ensuremath{\dagger}}
\newunicodechar{♯}{\ensuremath{\mathrel{\#}}}

\newunicodechar{∞}{\ensuremath{\infty}}
\newunicodechar{£}{\ensuremath{\mathrel{\$}}}

\newunicodechar{⧺}{\ensuremath{\doubleplus}}

\newcommand\xlrsquigarrow{
  \mathrel{%
    \vcenter{%
      \hbox{%
        \begin{tikzpicture}
          \path[
            draw,
            >={Implies[]},
            <->,
            double distance between line centers=1.5pt,
            decorate,
            decoration={
              zigzag,
              amplitude=0.7pt,
              segment length=3pt,
              pre length=4pt,
              post length=4pt,
            },
          ]   
            (0,0) -- (14pt,0);
        \end{tikzpicture}%
      }%
    }%
  }%
}

\usepackage[hidelinks,breaklinks]{hyperref} 
\usepackage[capitalize,noabbrev]{cleveref} 
\usepackage{url}
 
\usepackage{breakurl}

\usepackage[numbers]{natbib}

\usepackage{xparse}

\NewDocumentCommand{\optionalParens}{s m m}{
    \IfBooleanTF{#2}{(#3)}{\IfBooleanTF{#1}{~#3}{#3}}
}

\NewDocumentCommand{\apply}{m O{} s m t!}{
    \IfBooleanTF{#5}
    {#1#2}
    {#1#2 \optionalParens*{#3}{#4}}
}
\newcommand{\name}[1]{\apply{\textsf{#1}}}

\NewDocumentCommand{\applytwo}{m O{} s m s m}{ 
   #1#2~\optionalParens{#3}{#4}~\optionalParens{#5}{#6}
}

\def\THEN{\kwfont{then}}
\def\ELSE{\kwfont{else}}
\NewDocumentCommand{\ifThenElse}{m m m}{
  \apply{\kwfont{if}}{#1}~\apply{\THEN}{#2}~\apply{\ELSE}{#3}
}
\NewDocumentCommand{\while}{m m}{
  \apply{\kwfont{while}}{#1}~\apply{\kwfont{do}}{#2}
}

\newcommand{\IN}{\apply{\kwfont{in}}}
\NewDocumentCommand{\letin}{m m m t!}{
  \kwfont{let}~#1 ≔ #2 \IN{#3}
}

\usepackage[dvipsnames]{xcolor}

\definecolor{ltblue}{rgb}{0,0.4,0.4}
\definecolor{dkblue}{rgb}{0,0.1,0.6}
\definecolor{dkgreen}{rgb}{0,0.35,0}
\definecolor{dkviolet}{rgb}{0.3,0,0.5}
\definecolor{dkred}{rgb}{0.5,0,0}

\usepackage{wrapfig}

\usepackage{mathpartir}

  \usepackage{thmtools}
  \usepackage{thm-restate}

  \declaretheorem[name=Theorem,numberwithin=section
                 ,refname={theorem,theorems}
                 ,Refname={Theorem,Theorems}        ]
                 {theorem}
  \declaretheorem[sibling=theorem]{lemma}
  \declaretheoremstyle[headfont=\normalfont]{normalhead}
  
  \newtheorem{corollary}[theorem]{Corollary}
  \newtheorem{definition}[theorem]{Definition}
  
\newtheorem{proposition}[theorem]{Proposition}

\newtheorem*{theorem*}{Theorem}
\newtheorem*{proposition*}{Proposition}
\newtheorem*{definition*}{Definition}
\newtheorem*{lemma*}{Lemma}
\newtheorem*{corollary*}{Corollary}
\newtheorem*{axiom*}{Axiom}

\newtheorem*{thesis*}{Thesis Statement}

\crefname{procedure}{Procedure}{Procedures}

\crefname{question}{Question}{Questions}

\usepackage{enumerate}
\usepackage{paralist}

\usepackage{booktabs}

\usepackage{listings}
\usepackage[newfloat,frozencache]{minted} 
\usepackage{caption}

\usepackage{graphicx}

\usepackage{tikz} 
\usetikzlibrary{cd}
\usetikzlibrary{positioning}
\usetikzlibrary{%
  arrows,%
  shapes.misc,
  shapes.arrows,%
  shapes.callouts,
  chains,%
  matrix,%
  positioning,
  scopes,%
  decorations.pathmorphing,
  decorations.text,
  shadows%
}

\usepackage{makecell}

\newcommand{\kwfont}[1]{\textup{\texttt{#1}}} 

\usepackage{textcomp} 

\usepackage{xspace}

\usepackage{cmll} 

\usepackage{braket}
\NewDocumentCommand{\vectwo}{smm}{
    \IfBooleanTF{#1}
    {\begin{pmatrix} #2 \\ #3 \end{pmatrix}}
    {\left( \begin{smallmatrix} #2 \\ #3 \end{smallmatrix} \right)}
}
\NewDocumentCommand{\rowtwo}{smm}{
    \IfBooleanTF{#1}
    {\begin{pmatrix} #2 & #3 \end{pmatrix}}
    {\left( \begin{smallmatrix} #2 & #3 \end{smallmatrix} \right)}
}

\usepackage[nomargin,inline]{fixme}
\usepackage{fullpage}

\fxsetup{theme=color, status=draft}
\FXRegisterAuthor{sm}{asm}{\color{green!50!black}SM}
\FXRegisterAuthor{jp}{ajp}{\color{blue}JP}
\FXRegisterAuthor{rw}{arw}{\color{red}RW}
\FXRegisterAuthor{lm}{alm}{\color{purple}LM}
\FXRegisterAuthor{br}{abr}{\color{orange}BR}
\FXRegisterAuthor{ab}{aab}{\color{pink}AB}

\def\sourceColor{RoyalBlue}
\def\targetColor{Orange}
\def\intermediateColor{Green}
\def\hardwareColor{Purple}

\newcommand{\source}[1]{\textsf{\textcolor{\sourceColor}{#1}}\xspace}
\newcommand{\target}[1]{\textbf{\textcolor{\targetColor}{#1}}\xspace}
\newcommand{\intermediate}[1]{\textit{\textcolor{\intermediateColor}{#1}}\xspace}
\newcommand{\hardware}[1]{\textit{\textcolor{\hardwareColor}{#1}}\xspace}

\newcommand{\mathsource}[1]{ {\color{\sourceColor}{#1}} }
\newcommand{\mathtarget}[1]{\mathbf{\color{\targetColor}{#1}}\xspace}
\newcommand{\mathintermediate}[1]{\mathit{\color{\intermediateColor}{#1}}\xspace}
\newcommand{\mathhardware}[1]{\mathit{\color{\hardwareColor}{#1}}\xspace}

\newcommand{\BEHAVIOR}{\ensuremath{B}}
\newcommand{\behav}{\name{\BEHAVIOR}*}
\newcommand{\sbehav}{\name{\textcolor{\sourceColor}{\BEHAVIOR}}*}
\newcommand{\tbehav}{\name{\textcolor{\targetColor}{\BEHAVIOR}}*}
\newcommand{\ibehav}{\name{\textcolor{\intermediateColor}{\BEHAVIOR}}*}

\newcommand{\sourcecolor}{\source{blue}, \source{sans-serif}}
\newcommand{\targetcolor}{\target{orange}, \target{bold}}
\newcommand{\neutralcolor}{\textit{black}, \textit{italic}}

\NewDocumentCommand{\compile}{o}{
  \IfValueTF{#1}{#1 \! \downarrow}
                { \downarrow }
}
\newcommand{\ccompile}[1]{\ll \!\! #1 \! \gg}
\NewDocumentCommand{\compileFromTo}{o m m}{ \compile[#1]^{#2}_{#3} }

\NewDocumentCommand{\brelates}{O{} m m}{
  #2 \mapsto#1 #3
}
\NewDocumentCommand{\nbrelates}{O{} m m}{
  #2 \not\mapsto#1 #3
}

\newcommand{\IFSM}{\source{\textsf{IFSM}}}
\newcommand{\CPU}{\target{\textsf{CPU}}}

\newcommand{\Attack}{𝒜}

\newcommand{\OUTPUT}{\apply{\textsf{OUTPUT}}}

\newcommand{\IMP}{\intermediate{\textsf{IMP}}}

\newcommand{\ToyC}{\mathsf{\mathsource{{Toy}^{C}}}}
\newcommand{\ToyA}{\mathtarget{\textsf{Toy}^{\textsf{A}}}}
\newcommand{\ToyH}{\mathhardware{\textsf{Toy}^{\textsf{H}}}}
\newcommand{\LOne}{\mathsource{1}}
\newcommand{\LTwo}{\mathintermediate{2}}
\newcommand{\LThree}{\mathtarget{3}}

\newcommand{\SKIP}{\textsf{SKIP}}

\newcommand{\RETURN}{\textsf{RETURN}}
\newcommand{\HALT}{\textsf{HALT}}
\newcommand{\jmpz}{\applytwo{\textsf{JMPZ}}}

\newcommand{\dom}{\apply{\textsf{dom}}}

\newcommand{\overstep}[1]{\overset{#1}{→}}
\newcommand{\oversteps}[1]{\overset{#1}{→^∗}}

\newcommand{\true}{\textsf{true}}
\newcommand{\false}{\textsf{false}}

\newcommand{\DO}{\textsf{do}}
\newcommand{\AS}{\textsf{as}}

\NewDocumentCommand{\doAs}{s m m m}{
  \IfBooleanTF{#1}{
      \apply{\DO}*{#2}~\AS~#3.#4
    }{
      \apply{\DO}{#2}~\AS~#3.#4
    }}

\newcommand{\java}{\mintinline{java}}
\newcommand{\C}{\mintinline{c}}

\newminted{c}{linenos,frame=single,fontsize=\footnotesize,numbersep=2pt,escapeinside=@@}
\newminted{java}{linenos,frame=single,fontsize=\footnotesize,numbersep=2pt}
\newenvironment{ottdefnblock}[3][]{ \framebox{\mbox{#2}} \quad #3 \\[0pt]}{}

\newcommand{\ottnt}[1]{\mathit{#1}}
\newcommand{\ottmv}[1]{\mathit{#1}}
\newcommand{\ottkw}[1]{\mathbf{#1}}
\newcommand{\ottsym}[1]{#1}

\newenvironment{IMPdefnblock}[3][]{ \framebox{\mbox{#2}} \quad #3 \\[0pt]}{}

\newenvironment{Cdefnblock}[3][]{ \framebox{\mbox{#2}} \quad #3 \\[0pt]}{}

\begin{document}

\title{{ \normalsize
Distribution Statement ``A'' (Approved for Public Release, Distribution Unlimited)} \\ Weird Machines as Insecure Compilation\thanks{This material is based upon work supported by the United States Air Force and DARPA under Contract No. FA8750-15-C-0124.  Any opinions, findings and conclusions or recommendations expressed in this material are those of the author(s) and do not necessarily reflect the views of the United States Air Force and DARPA.}
}

\author{
  \IEEEauthorblockN{Jennifer Paykin
                    Eric Mertens,
                    Mark Tullsen,
                    Luke Maurer, \\
                    Beno\^{i}t Razet,
                    Alexander Bakst, and
                    Scott Moore
                   }

\IEEEauthorblockA{
  Galois, Inc. \\
  Portland, OR, USA 
}
\IEEEauthorblockA{
  first.last@galois.com
}
}

\maketitle

\begin{abstract}

Weird machines---the computational models accessible by exploiting security vulnerabilities---arise from the difference between the model a programmer has in her head of how her program should run and the implementation that actually executes. Previous attempts to reason about or identify weird machines have viewed these models through the lens of formal computational structures such as state machines and Turing machines. But because programmers rarely think about programs in this way, it is difficult to effectively apply insights about weird machines to improve security.

We present a new view of weird machines based on techniques from programming languages theory and secure compilation. Instead of an underspecified model drawn from a programmers' head, we start with a program written in a high-level source language that enforces security properties by design. Instead of state machines to describe computation, we use the well-defined semantics of this source language and a target language, into which the source program will be compiled. Weird machines are the sets of behaviors that can be achieved by a compiled source program in the target language that cannot be achieved in the source language directly. That is, exploits are witnesses to insecure compilation.

This paper develops a framework for characterizing weird machines as insecure compilation, and illustrates the framework with examples of common exploits. We study the classes of security properties that exploits violate, the compositionality of exploits in a compiler stack, and the weird machines and mitigations that arise.

\end{abstract}


\section{Introduction}

Exploits serve an important role in security research: they witnesses the insecurity of a system by causing it to behave inappropriately.
However, an exploit by itself fails to answer many important questions about the system under consideration.
How severe is the vulnerability in question?
Does it represent a new class of attacks, or is it a member of a known class of attacks?
Is the vulnerability an implementation error, a design flaw, or an emergent property of the entire system?
Can it be repaired and if so, how effective is a proposed mitigation?
Without a systematic approach to understanding exploits, it is difficult to evaluate the importance of any particular vulnerability or to generalize lessons learned to improve security more broadly.

In the exploit community, practitioners describe hacking as an exercise in ``programming the \emph{weird machine}'' of a system. A weird machine is the latent computational machine exposed by a vulnerable program that can be repurposed by an attacker to achieve their goals~\citep{Bratus2009,Bratus2011,Bratus2017}. One of the most stark examples of weird machines in practice is return-oriented programming (ROP), where attackers exploit a program by overwriting the stack with a sequence of return addresses that invoke fragments of the original binary to achieve a desired effect~\citep{Shacham2007}.

Despite the intuitive appeal of weird machines, it has proven challenging to provide a formal definition that can be consistently applied to a variety of systems and vulnerabilities~\citep{Vanegue2014,Dullien2018}.
\citet{Dullien2018} defines weird behavior as the difference between two state machines---an \emph{intended finite state machine} (IFSM) corresponding to the model that the programmer has in her head when writing the program and an
implementation that attempts to realize the IFSM. A weird state in the implementation is one that does not correspond to a state in the IFSM. A weird machine is the collection of computations reachable from a weird state.

But state machines obscure the abstractions present in the source program, which might be enforced by language features or data structures.
Indeed, \citet{Bratus2017} argue that weird machines arise exactly when an attack causes a program to encounter an unexpected program state by violating the program's expected abstractions.  Furthermore, the state machine approach does not explain how the IFSM and the implementation are related, or how to know if the implementation is sound. 

In the study of programming languages, the relationship between a high-level representation of a program and a low-level implementation is given as a \emph{compiler}. The field of \emph{secure compilation}~\citep{Abadi1999,Patrignani2018,Abate2019} studies when and how security properties are preserved by a compiler. Historically, secure compilation referred to a compiler satisfying \emph{full abstraction}, the property that two programs are equivalent in the source language if and only if their compilations are equivalent in the target language. If a compiler is fully abstract, attacks on the compiled version of the program should only be as strong as attacks on the original high-level version, and so it is sufficient for a programmer to reason about security properties at the source code level, without worrying about insecurity introduced during compilation. More recently, researchers have proposed \emph{robust property preservation} as an alternative characterization of a compiler security that says a compiler is secure if it preserves the security of behaviors of source programs~\citep{Abate2019,Patrignani2017}.

Our contribution in this work is to provide a new formal definition of  weird machines based on secure compilation and a framework for reasoning about programs, vulnerabilities, and exploits. We define the exploits of a vulnerable source-language component $\source{V}$ to be the target-language attacker contexts $\target{A}$ that are counterexamples to secure compilation between $\source{V}$ and its compilation $ \compile[  \source{V}  ] $. Formally, the behavior of an exploit linked with this compilation cannot be simulated by any context in the source language. The weird machine of $\source{V}$ is the class of behaviors arising from exploits of $\source{V}$.\footnotemark

\footnotetext{
    Throughout the paper we use \sourcecolor{} text to refer to meta-variables and notations in the source language, \targetcolor{} text for meta-variables and notations in the target language, and \neutralcolor{} text for meta-variables and notations of either language. We suggest viewing or printing the paper in color for maximum readability.
}

With the framework that arises,
we model a number of weird machines and exploits including return-oriented and data-oriented programming~\citep{Shacham2007,Chen2005}, access control in Java~\citep{Abadi1999}, information flow control~\citep{Sabelfeld2003}, and timing side-channel attacks (\cref{sec:mainideas,sec:robust}).
We prove that exploits are exactly the contexts that violate robust properties of behaviors~\citep{Abate2019} and identify several other classes of exploits that violate sub-classes of security properties (\cref{sec:robust}). 

We prove that exploits propagate up and down a compiler stack under certain conditions, so exploits in a small part of a compiler stack can be understood as exploits of a larger system~(\cref{sec:compositionality}).
We show that our approach generalizes Dullien's state machine formalization, but focuses on the behavior of the attack rather than the mechanism of the attack~(\cref{sec:dullien}). For example, C programs using undefined behavior in a safe way are not considered exploits in our framework, but are in the prior work.
Finally, we discuss how our framework can highlight the expressiveness of certain exploits, their underlying causes, and their mitigations~(\cref{sec:discussion}).

\section{Main ideas}
\label{sec:mainideas}
In this section we develop a formal framework of exploits and weird machines. This framework considers two programming languages---a source language $ \source{Source} $ and a target language $ \target{Target} $---connected via a compiler $ \compile $. The pieces of this framework, described below, are summarized in \cref{fig:framework}.



\begin{figure}
\small
\begin{tabular}{lll}
  \toprule
  language          & $ \source{Source} $      & $ \target{Target} $  \\
  \midrule
  whole programs    & $\source{P}^\source{S}$         & $\target{P}^\target{T}$     \\
  behaviors         & $\source{b}^\source{S}  \in   \mathsource{ \mathcal{B} }^\source{S} $
                    & $\target{b}^\target{T}  \in   \mathtarget{ \mathcal{B} }^\target{T} $    \\
  whole program semantics
                    & $ \sbehav{ \source{P}^\source{S} }   \triangleq  \source{b}^\source{S}$
                    & $ \tbehav{ \target{P}^\target{T} }   \triangleq  \target{b}^\target{T}$    \\
  components        & $\source{U}^\source{S}$           & $\target{U}^\target{T}$       \\
  contexts          & $\source{C}^\source{S}$        & $\target{C}^\target{T}$    \\
  linking           & $\source{C}^\source{S}  \ottsym{[}  \source{U}^\source{S}  \ottsym{]} ≜ \source{P}^\source{S}$
                    & $\target{C}^\target{T}  \ottsym{[}  \target{U}^\target{T}  \ottsym{]} ≜ \target{P}^\target{T}$             \\
  attacker          &                   & $\target{A}  \in   \mathtarget{\Attack}  ⊆ \{\target{C}^\target{T}\}$ \\
  \bottomrule \\
  \toprule
  compiler  &   \multicolumn{2}{l}{$ \compileFromTo{  \source{Source}  }{  \target{Target}  } $ (or just $ \compile $)} \\
  \midrule
  behavior relation & \multicolumn{2}{l}{$ \brelates{ \source{b}^\source{S} }{ \target{b}^\target{T} } $} \\
  whole programs    & \multicolumn{2}{l}{$ \compile[  \source{P}^\source{S}  ]  ≜ \target{P}^\target{T}$} \\
  components        & \multicolumn{2}{l}{$ \compile[  \source{U}^\source{S}  ]  ≜ \target{U}^\target{T}$} \\
  contexts          & \multicolumn{2}{l}{$ \compile[  \source{C}^\source{S}  ]  ≜ \target{C}^\target{T}$} \\
  modularity        & \multicolumn{2}{l}{$ \tbehav{   \compile[  \source{C}^\source{S}  ]  [   \compile[  \source{U}^\source{S}  ]   ]  } $ 
                      $=  \tbehav{  \compile[   \source{C}^\source{S} [  \source{U}^\source{S}  ]   ]  } $} \\
  \bottomrule
\end{tabular}
\caption{Summary and notation of the framework.}
\label{fig:framework}
\end{figure}

\subsubsection*{Behaviors, contexts and components}

We assume each language is associated with a set of \emph{behaviors} $ \mathcal{B} $, which we denote with the meta-variables $\ottnt{b}  \in   \mathcal{B} $. The behavior of a \emph{whole program} $P$, written $ \behav{ P } $, describes its semantics. For example:
\begin{itemize}
    \item In C, the behavior of a program $P$ is a set of input-output traces corresponding to possible executions of the program.
    \item In functional languages, the behavior of an expression is its value  and/or termination behavior.
    \item In SQL (\cref{example:SQL}), the behavior of a query $P$ is a function from a relational database to the table corresponding to the result of the query.
\end{itemize}
Frequently, but not always, the set of behaviors has the form $ \mathcal{B} = \mathcal{P}(  \textsf{Trace}  ) $, where $ \textsf{Trace} $ is a set of traces.

In practice programs are often modular, making use of external libraries or relying on input from untrusted sources (e.g. a server accepting input or a program relying on a config file). Exploits arise when some piece of a program $U$, called a \emph{component}, is \emph{linked} with an adversarial surrounding environment $C$, called a \emph{context}. This linking operation forms a whole program $C  \ottsym{[}  U  \ottsym{]}$. 
\begin{itemize}
    \item In C, a component is a compilation unit that is linked with a context comprising other compilation units (perhaps including a \mintinline{c}{main} procedure).
    \item In Java, components and contexts are sets of classes that refer to each other.
  \item In functional languages, a component is an open expression and a context is an expression with a hole that can be filled with the component.
    \item In languages with I/O, the context may also include an environment that provides input to a standalone program component.
\end{itemize}
The result of linking may not always be defined, 
but for the sake of simplicity we assume that whenever $C  \ottsym{[}  U  \ottsym{]}$ occurs in this paper, it is well-defined.


\subsubsection*{Compilers}

A compiler maps programs in the source language $ \source{Source} $ to programs in the target langauge $ \target{Target} $. For a whole source program $\source{P}^\source{S}$ we write $ \compile[  \source{P}^\source{S}  ] $ for its compilation. We say that a compiler is \emph{modular} if compilation is also defined on source components $\source{U}^\source{S}$ and source contexts $\source{C}^\source{S}$ such that $ \tbehav{  \compile[   \source{C}^\source{S} [  \source{U}^\source{S}  ]   ]  } = \tbehav{   \compile[  \source{C}^\source{S}  ]  [   \compile[  \source{U}^\source{S}  ]   ]  } $. Unless otherwise specified, assume that all the compilers occurring in this paper are modular.

Understanding the effect of compilation requires relating behaviors in the source language ($\source{b}^\source{S}$) with behaviors in the target language ($\target{b}^\target{T}$); this relation is written $ \brelates{ \source{b}^\source{S} }{ \target{b}^\target{T} } $ and pronounced ``$\source{b}^\source{S}$ maps to $\target{b}^\target{T}$''. When the source and target behaviors are drawn from the same set ($ \mathsource{ \mathcal{B} }^\source{S}  =  \mathtarget{ \mathcal{B} }^\target{T} $), this relation may just be equality, but in general we will assume very little structure on the relation itself. To see why, consider the following:
\begin{itemize}

    \item $ \mapsto $ need not be functional: Suppose the source language has nondeterministic behaviors, written $\source{b}^\source{S}_{{\mathrm{1}}} ⊕ \source{b}^\source{S}_{{\mathrm{2}}}$, and the target language does not, so that $\brelates{\source{b}^\source{S}_{{\mathrm{1}}} ⊕ \source{b}^\source{S}_{{\mathrm{2}}}}{\target{b}^\target{T}_{{\mathrm{1}}}}$ and $\brelates{\source{b}^\source{S}_{{\mathrm{1}}} ⊕ \source{b}^\source{S}_{{\mathrm{2}}}}{\target{b}^\target{T}_{{\mathrm{2}}}}$.

    \item $ \mapsto $ need not be well-defined on all inputs: in C (as in \cref{sec:undefined}), the undefined behavior $ \source{undef} $ is not related to any target behaviors: $\forall  \target{b}^\target{T}  .\xspace   \nbrelates{  \source{undef}  }{ \target{b}^\target{T} } $.

    \item $ \mapsto $ need not be injective: If the source language has both boolean values and natural numbers and the target language has only integers (as in \cref{sec:undefined}), we will have both $\brelates{\mathsource{\true}}{\mathtarget{1}}$ and
$\brelates{\mathsource{1}}{\mathtarget{1}}$.

    \item $ \mapsto $ need not be surjective: As above, if the source language has booleans and natural numbers but no (negative) integers, then there is no $\source{b}^\source{S}$ such that $\brelates{\source{b}^\source{S}}{\mathtarget{-1}}$.

\end{itemize}
When the behaviors of the source and target language each defined as a set of traces ($ \mathsource{ \mathcal{B} }^\source{S}   \ottsym{=}   \mathcal{P}(  \source{ \textsf{Trace} }^\source{S}  ) $ and $ \mathtarget{ \mathcal{B} }^\target{T}   \ottsym{=}   \mathcal{P}(  \target{ \textsf{Trace} }^\target{T}  ) $), then the relation $ \brelates{ \source{b}^\source{S} }{ \target{b}^\target{T} } $ can be decomposed into a relation $ \brelates{ \source{t}^\source{S} }{ \target{t}^\target{T} } $ on traces (reusing notation):
\small
\begin{align*}
    & \brelates{ \source{b}^\source{S} }{ \target{b}^\target{T} }  ⇔ \\
        &\left(\forall  \source{t}^\source{S}  \in  \source{b}^\source{S}  .\xspace  \exists\xspace  \target{t}^\target{T}  \in  \target{b}^\target{T}  .\xspace   \brelates{ \source{t}^\source{S} }{ \target{t}^\target{T} } \right)
      ∧ \left(\forall  \target{t}^\target{T}  \in  \target{b}^\target{T}  .\xspace  \exists\xspace  \source{t}^\source{S}  \in  \source{b}^\source{S}  .\xspace   \brelates{ \source{t}^\source{S} }{ \target{t}^\target{T} } \right)
\end{align*}
\normalsize


\subsubsection*{Attackers}

A security exploit consists of two parts: a vulnerable source component $\source{V}$ and an attacking target context $\target{A}$ drawn from some set $ \mathtarget{\Attack} $. The set $ \mathtarget{\Attack} $ corresponds to the sorts of attacks being considered.
For example, in the case of data-oriented programming, $ \mathtarget{\Attack} $ may allow arbitrary writes of constants to memory  but cannot affect control flow. In the theoretical language with I/O, $ \mathtarget{\Attack} $ may be allowed to perform input/output with the source component but not actually execute any code.

\begin{figure}
\centering
\includegraphics[width=\columnwidth]{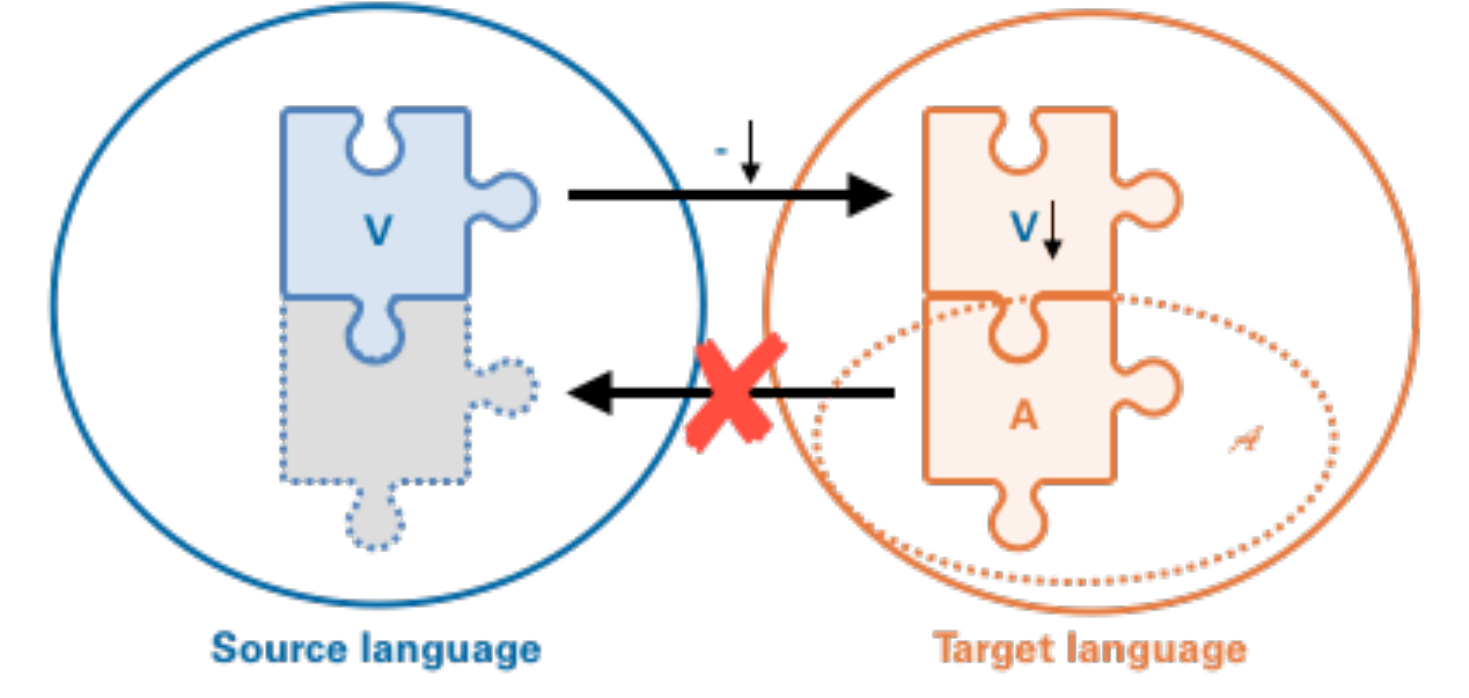}
\caption{An attacker context $\target{A}  \in   \mathtarget{\Attack} $ is in the weird machine of  $\source{V}$ if and only if there is no back-translated context $\source{C}^\source{S}$ such that $\source{C}^\source{S}  \ottsym{[}  \source{V}  \ottsym{]}$ behaves the same as $\target{A}  \ottsym{[}   \compile[  \source{V}  ]   \ottsym{]}$.}
\label{fig:SourceTargetLanguages}
\end{figure}

\subsection{Exploits and weird machines}

An attack $\target{A}  \in   \mathtarget{\Attack} $ invokes \emph{weird behavior} if it interacts with $ \compile[  \source{V}  ] $ in a way that cannot be achieved in the source language. In other words, a context is only an exploit if there is no source context $\source{C}^\source{S}$ such that the behavior of $\source{C}^\source{S}  \ottsym{[}  \source{V}  \ottsym{]}$ is related to the behavior of $\target{A}  \ottsym{[}   \compile[  \source{V}  ]   \ottsym{]}$. This relationship is illustrated in \cref{fig:SourceTargetLanguages}.

\begin{definition} \label{def:WM}
    An \emph{exploit} of a vulnerable source program $\source{V}$ is a target attacker context $\target{A}  \in   \mathtarget{\Attack} $  such that there is no source context $\source{C}^\source{S}$ with $ \brelates{  \sbehav{  \source{C}^\source{S} [  \source{V}  ]  }  }{  \tbehav{  \target{A} [   \compile[  \source{V}  ]   ]  }  } $.
    \[
          \textsf{Exploit}  ^{  \mathtarget{\Attack}  }   \ottsym{(}  \source{V}  \ottsym{)}  \triangleq   \left\{  \target{A}  \in   \mathtarget{\Attack}   \mid  \neg  \exists\xspace  \source{C}^\source{S}  .\xspace   \brelates{  \sbehav{  \source{C}^\source{S} [  \source{V}  ]  }  }{  \tbehav{  \target{A} [   \compile[  \source{V}  ]   ]  }  }   \right\} 
    \]

  The \emph{weird machine} of a vulnerable source program $\source{V}$ is the collection of behaviors arising from exploits of $\source{V}$.
    \[
          \textsf{WM}  ^{  \mathtarget{\Attack}  }   \ottsym{(}  \source{V}  \ottsym{)}  \triangleq   \left\{   \tbehav{  \target{A} [   \compile[  \source{V}  ]   ]  }   \mid  \target{A}  \in    \textsf{Exploit}  ^{  \mathtarget{\Attack}  }   \ottsym{(}  \source{V}  \ottsym{)}  \right\} 
    \]

    When $ \mathtarget{\Attack} $ can be inferred from the context, we will write $ \textsf{Exploit}   \ottsym{(}  \source{V}  \ottsym{)}$ and $ \textsf{WM}   \ottsym{(}  \source{V}  \ottsym{)}$ in place of $  \textsf{Exploit}  ^{  \mathtarget{\Attack}  }   \ottsym{(}  \source{V}  \ottsym{)}$ and $  \textsf{WM}  ^{  \mathtarget{\Attack}  }   \ottsym{(}  \source{V}  \ottsym{)}$ respectively.
\end{definition}


\subsection{Example: SQL injection attack}
\label{example:SQL}

Consider a source language made up of structured SQL queries. The following program allows the user to access information associated with a particular user id:
\begin{ccode}
uid <- getIntegerInput();
SELECT * FROM Users WHERE UserId = uid
\end{ccode}
This code is then compiled to a target language where the SQL query is implemented as a string, processed by an external call to a database.
\begin{ccode}
uid <- getStringInput();
queryDB("SELECT * FROM Users \
         WHERE UserId = " + uid);
\end{ccode}
A SQL injection attack interacts with the compiled version of the code by setting the input \C{uid} to the string \C{"0 OR true"}. This results in a behavior (outputs the entire \mintinline{c}{Users} table) that the source program (which only ever outputs a single user's information) cannot achieve.

\subsection{Example: Finite memory}
\label{example:finite-memory}

Consider a source language with unbounded heap size, and a target language with a maximum heap of size $n$. We define the behavior of programs in both the source and the target languages to be the maximum size of the heap used in execution of the program. The target language can also throw an error, \mintinline{c}{out_of_resources} if the program attempts to allocate more than $n$ references, and there is no source behavior such that $\brelates{n}{\mintinline{c}{out_of_resources}}$. 
Thus, if $\source{V}$ is a source component such that $\target{A}  \ottsym{[}   \compile[  \source{V}  ]   \ottsym{]}$ allocates more than $n$ memory references, then $\target{A}  \in   \textsf{Exploit}   \ottsym{(}  \source{V}  \ottsym{)}$.

\subsection{Example: Access control}
\label{example:access-control}

As recently as Java 8, Java had a security vulnerability whereby attackers could gain access to a class \java{C}'s private field if \java{C} also had an inner class that accessed the field. In Java, inner classes have access to all of their outer class's private data, but the JVM does not have any notion of inner classes. The result is that the private data accessed by the inner class is exposed publicly, as illustrated by \cref{lst:java-example} (adapted from~\citep{Abadi1999}). In this example, source contexts of \java{OuterClass} cannot change the value of \java{flag}, but target level contexts can modify it by invoking the synthesized accessor \java{access$002}.

The Java compiler ensures that attacks of the form $ \compile[  \source{A}  ] $ (where $\source{A}$ is a source Java program) will not be able to access this method, but there is nothing stopping an attacker from manually constructing a JVM class that uses it directly. This class can subvert any security properties of the source program that rely on the privacy of the \mintinline{java}{flag} field.

\begin{listing}
\begin{javacode}
public class OuterClass {
  private boolean flag = false;
  private class InnerClass {
    void setFlag(){
      flag = true;
    }
  }
}
\end{javacode}
\begin{javacode}
public class OuterClass
  extends java.lang.Object{
public OuterClass();
  Code: /* ... */
static boolean access$002(OuterClass, boolean);
  Code:
   0:	aload_0
   1:	iload_1
   2:	dup_x1
   3:	putfield  #1; //Field flag:Z
   6:	ireturn
}
public class OuterClass$InnerClass
  extends java.lang.Object{ /* ... */ }
\end{javacode}
\caption{A Java program (top) uses an inner class that when compiled (bottom) will expose private data.}
\label{lst:java-example}
\end{listing}

\subsection{Undefined behavior}
\label{sec:undefined}

Some of the most well-known exploits derive from the use of \emph{undefined behavior} in C. For example, accessing memory outside of the bounds of an array is undefined according to the C standard~\citep{C18}, which means that the language provides \emph{no} guarantees about the behavior of a compiled program during an execution that would have performed such an access. In practice, rather than emitting code with an arbitrary behavior, a C compiler will emit code with predictable (if surprising) behavior, for example compiling an out-of-bounds read to either
\begin{inparaenum}
\item \label{it:silent} silently read from memory outside of the bounds of the array;
\item throw an error when such an out-of-bounds access occurs; or
\item enter an infinite loop; etc.
\end{inparaenum}
Given a particular compiler, an exploit writer can observe the choice made by the compiler (in practice, almost always \ref{it:silent}) and use the actual behavior of the compiled code for her own purposes. By understanding the compilation strategies a compiler makes for undefined behaviors and the layout of the stack, attackers can craft exploits from syntactically valid but semantically undefined C code fragments. That is, exploits of C programs $\source{V}$ often have the form $ \compile[  \source{A}  ] $ where $\source{A}$ is a C program that gives rise to undefined behaviors when linked with $\source{V}$.

To formalize this, assume there is some C trace behavior $ \source{undef}   \in   \source{ \textsf{Trace} }^\source{S} $. The C compiler CompCert~\citep{leroy2009} in correct in the sense that, if $ \source{undef}   \notin   \sbehav{ \source{P}^\source{S} } $, then $\target{t}^\target{T}  \in   \tbehav{  \compile[  \source{P}^\source{S}  ]  } $ implies that there is som $\source{t}^\source{S}  \in   \sbehav{ \source{P}^\source{S} } $ such that $ \brelates{ \source{t}^\source{S} }{ \target{t}^\target{T} } $.
However, CompCert does not make any claims about the behavior of $ \compile[  \source{P}^\source{S}  ] $ when $\source{t}^\source{S}  \in   \sbehav{ \source{P}^\source{S} }  =  \source{undef} $.
To account for undefined behavior, we say that whenever $ \brelates{  \source{undef}  }{ \target{t}^\target{T} } $, it must either be the case that $\target{t}^\target{T}= \target{undef} $ or there is no $ \target{undef}   \in   \target{ \textsf{Trace} }^\target{T} $, in which case there is no such $\target{t}^\target{T}$.

\subsection{Example: Return-oriented programming}
\label{example:ROP}

\begin{listing}[t]
\begin{ccode}
void echo(void) {
  char[64] buf;
  gets(buf);
  printf("%s\n", buf);
}

int main() {
  while (1) { echo(); }
}
\end{ccode}
\caption{A program vulnerable to return-oriented programming.}
\label{lst:echo}
\end{listing}

Consider the simple `echo' C program in \cref{lst:echo}. If an attacker provides more than 64 characters to (the compilation of) \mintinline{c}{echo}, the data will overflow onto the rest of the stack, overwriting the return pointer. When the program returns from \mintinline{c}{echo}, instead of returning to the calling function, it will return to the new address in the return pointer, either causing a segfault or continuing with some other code. By carefully choosing which address to write to the return pointer and other locations on the stack, an attacker can launch a return-oriented programming attack on the program, executing a malicious payload.

It is easy to see that return-oriented programming can result in behaviors that were unachievable by a source context, such as causing the program to launch a shell or invoke other system calls.
If the context linked with a component can run arbitrary code, then this behavior may not suffice to show the attack is an exploit, since the source context could have achieved the same effect. Even in this scenario, however, some uses of return-oriented programming will be identified as exploits because they use altered control flow to violate other abstractions in the source program.

\subsection{Exploits violate language abstractions}
\label{sec:secrets}

In order to identify security violations, we use the source language semantics as an oracle for secure behavior.
That is, an attack is only an exploit if it violates some assumptions or \emph{abstractions} of the source language. 

\begin{listing}[t]
\begin{ccode}
/* auth.h */
bool auth();
char guess[8];
\end{ccode}



\begin{ccode}
/* auth.c */
#include "auth.h"

static char secret[8];

static void init_secret() {
  strcpy(secret, "weird");
}

void auth() {
  init_secret();
  if (0 != strcmp(secret, guess)) {
    printf("Wrong password: %s", guess); @\label{line:printsecret}@
    exit(-1);
  } else {
    printf("Logged in");
  }
}
\end{ccode}
\caption{Authentication function with constant secret.}
\label{lst:auth}
\end{listing}

\begin{listing}[t]
\begin{ccode}
/* main.c */
#include "auth.h"

void before() { /* provided by attacker */ }
void after()  { /* provided by attacker */ }

int main() {
  before();
  auth();
  after();
}
\end{ccode}
\caption{A class of attacker contexts for the authentication function.}
\label{lst:authcontexts}
\end{listing}

\begin{listing}[t]
\begin{ccode}
void before() { strcpy(guess, "12345678"); }
void after()  { return; }
\end{ccode}
\caption{An attack that prints the password.}
\label{lst:attackpassword}
\end{listing}

For example, consider the C function \mintinline{c}{auth} in \cref{lst:auth} that implements an authentication mechanism using the hard-coded password \mintinline{c}{"weird"}. Suppose we wish to consider the security of this function against an attacker linking to it from a larger C program, modeled by the class of attacker contexts given in \cref{lst:authcontexts}.

The attack in \cref{lst:attackpassword} extracts the password using a carefully crafted input that writes its null terminator outside of the \mintinline{c}{guess} array (into the first element of the adjacent \mintinline{c}{secret}). But the null terminator for
\mintinline{c}{guess} is subsequently overwritten when \mintinline{c}{init_secret()} copies \mintinline{c}{"weird"} to \mintinline{c}{secret}. Thus the "string" in \mintinline{c}{guess} is suffixed with the contents of \mintinline{c}{secret}, which will be output on line~\ref{line:printsecret}, revealing the secret: \mintinline{shell}{Wrong password: 12345678weird}.

\begin{listing}[t]
\begin{ccode}
void before() { strcpy(guess, "1234567"); }
void after()  { printf("8weird"); }
\end{ccode}
\caption{A context with the same behavior as \cref{lst:attackpassword} in the source semantics.}
\label{lst:authattackequiv}
\end{listing}

Perhaps unintuitively, this context would \emph{not} be considered an exploit according to our definition, because the context in \cref{lst:authattackequiv} exhibits the same behavior according to the C semantics. The existence of the constant value \mintinline{c}{"weird"} is not protected by a language-level abstraction and so an attacker can easily simulate the behavior of an exploit that accesses the password.

Suppose instead that the password was hidden by a language-level abstraction, such as a different implementation of \mintinline{c}{init_secret} that loads the (possibly different on each execution) secret from some external source (e.g., an access-controlled configuration file). In that case, the context in \cref{lst:attackpassword} would be an exploit, since according to the source-language semantics, no context can access the static buffer \mintinline{c}{secret} or the static function \mintinline{c}{init_secret}.

The security properties of such a component $\source{V}$ can be broken into two parts:
\begin{inparaenum}
\item its security with respect to source language contexts; and 
\item its security with respect to target language contexts.
\end{inparaenum}
In general programmers should take perspective 1 into account when writing code and developing algorithms, since they are already in the mindspace of the source language. But programmers often don't know the details of a compiler or target language, so it is significantly harder to understand perspective 2.

In \cref{sec:relprop}, we discuss an extension of our framework that would identify the attack against \cref{lst:auth} as an exploit by also considering attacks that violate code confidentiality. This extended definition can account for additional exploits under an attacker model where the adversary is assumed not to have access to the program's source code, at the cost of increasing the difficulty of reasoning about program exploitability. However, even with this extension, there are still programs that are intuitively insecure but that our framework says are unexploitable, such as a program that always stores user's passwords in plaintext, violating well-known security best practices.

\section{Robust property violation}
\label{sec:robust}
One way to characterize classes of exploits is to consider what security properties of a source program they can violate. For example, return-oriented programming attacks can change the the control-flow of a program, while data-oriented programming attacks can only alter the data that a program computes on. Similarly, timing side channels may reveal confidential information, but cannot directly change the execution of a program.
In general, showing that an attacker context violates a property that holds of all source programs is sufficient to show that it has invoked a weird behavior.

For a property $\mathsource{ \mathbb{B} }  \subseteq   \mathsource{ \mathcal{B} }^\source{S} $ of source behaviors, we write $ \compile[  \mathsource{ \mathbb{B} }  ] $ for the set of target behaviors satisfying $\mathsource{ \mathbb{B} }$:
  \[
     \compile[  \mathsource{ \mathbb{B} }  ]  ≜  \left\{  \target{b}^\target{T}  \in   \mathtarget{ \mathcal{B} }^\target{T}   \mid  \exists\xspace  \source{b}^\source{S}  \in  \mathsource{ \mathbb{B} }  .\xspace   \brelates{ \source{b}^\source{S} }{ \target{b}^\target{T} }   \right\} .
  \]

\begin{lemma}
  Let $\mathsource{ \mathbb{B} }  \subseteq   \mathsource{ \mathcal{B} }^\source{S} $ be a property of source behaviors such that for all whole source programs $\source{P}^\source{S}$, $ \sbehav{ \source{P}^\source{S} }   \in  \mathsource{ \mathbb{B} }$. 
  If $\source{V}$ is a source component and $\target{A}  \in   \mathtarget{\Attack} $ is a target 
  context such that $ \tbehav{  \target{A} [   \compile[  \source{V}  ]   ]  }   \notin   \compile[  \mathsource{ \mathbb{B} }  ] $, then $\target{A}  \in    \textsf{Exploit}  ^{  \mathtarget{\Attack}  }   \ottsym{(}  \source{V}  \ottsym{)}$.
\end{lemma}
\begin{proof}
    Since $ \tbehav{  \target{A} [   \compile[  \source{V}  ]   ]  }   \notin   \compile[  \mathsource{ \mathbb{B} }  ] $, there is no source behavior $\source{b}^\source{S}  \in  \mathsource{ \mathbb{B} }$ such that $ \brelates{ \source{b}^\source{S} }{  \tbehav{  \target{A} [   \compile[  \source{V}  ]   ]  }  } $; so there is no context $\source{C}^\source{S}$ such that $ \brelates{  \sbehav{  \source{C}^\source{S} [  \source{V}  ]  }  }{  \tbehav{  \target{A} [   \compile[  \source{V}  ]   ]  }  } $.
\end{proof}

In fact, the violating property need not hold of all source programs, but only those of the form $\source{C}^\source{S}  \ottsym{[}  \source{V}  \ottsym{]}$.

\begin{lemma} \label{lem:WM-RHP1}
    Let $\source{V}$ be a source component, and let $\mathsource{ \mathbb{B} }  \subseteq   \mathsource{ \mathcal{B} }^\source{S} $ be a property of
    source behaviors such that for all source contexts $\source{C}^\source{S}$, it is the
    case that $ \sbehav{  \source{C}^\source{S} [  \source{V}  ]  }   \in  \mathsource{ \mathbb{B} }$.
    If $\target{A}  \in   \mathtarget{\Attack} $ is an attack such that 
    $ \tbehav{  \target{A} [   \compile[  \source{V}  ]   ]  }   \notin   \compile[  \mathsource{ \mathbb{B} }  ] $, then $\target{A}  \in   \textsf{Exploit}   \ottsym{(}  \source{V}  \ottsym{)}$.
\end{lemma}

This result derives from the literature on robust property preserving compilers~\citep{Patrignani2017,Abate2019,Abate2019a}. A compiler satisfies \emph{robust (hyper-) property preservation} (RHP) if, whenever a hyperproperty is preserved by a source program, the compilation of the hyperproperty is preserved by the compilation of the source program.\footnote{For a much deeper discussion of robust property preservation with respect to the bhavior relation, see \citet{Abate2019a}; prior to that, it was assumed that source and target behaviors were always equal.}
\begin{align*}
    ∀ \mathsource{ \mathbb{B} }  \subseteq   \mathsource{ \mathcal{B} }^\source{S} .~ ∀ \source{U}^\source{S}.
        &\left(∀ \source{C}^\source{S}.~  \sbehav{  \source{C}^\source{S} [  \source{U}^\source{S}  ]  }   \in  \mathsource{ \mathbb{B} }\right) ⇒ \\
        &\left(∀ \target{C}^\target{T}.~  \tbehav{  \target{C}^\target{T} [   \compile[  \source{U}^\source{S}  ]   ]  }   \in   \compile[  \mathsource{ \mathbb{B} }  ] \right)
    \tag{RHP}
\end{align*}
\citet{Abate2019} also give a property-free characterization of robust hyperproprety preservation: 
\begin{theorem}[\citet{Abate2019}]
    A compiler satisfies RHP if and only if for all source components $\source{U}^\source{S}$ and target contexts $\target{C}^\target{T}$, there exists a back-translated source context $\source{C}^\source{S}$ such that $ \brelates{  \sbehav{  \source{C}^\source{S} [  \source{U}^\source{S}  ]  }  }{  \tbehav{  \target{C}^\target{T} [   \compile[  \source{U}^\source{S}  ]   ]  }  } $.
\end{theorem}

This property-free characterization is exactly the negation of \cref{def:WM}, so we can restate the result in terms of exploits:

\begin{theorem}
    A compiler satisfies RHP if and only if it has no exploits: for all
    source components $\source{U}^\source{S}$, $ \textsf{Exploit}   \ottsym{(}  \source{U}^\source{S}  \ottsym{)} = ∅$.
\end{theorem}

\cref{lem:WM-RHP1} establishes one direction of this claim, and we present the other for completeness:

\begin{lemma} \label{lem:WM-RHP2}
    If $\target{A}  \in   \textsf{Exploit}   \ottsym{(}  \source{V}  \ottsym{)}$, then there exists some property $\mathsource{ \mathbb{B} }  \subseteq   \mathsource{ \mathcal{B} }^\source{S} $ such that
    $ \tbehav{  \target{A} [   \compile[  \source{V}  ]   ]  }   \notin   \compile[  \mathsource{ \mathbb{B} }  ] $ but for all source contexts $\source{C}^\source{S}$ it is the case that
    $ \sbehav{  \source{C}^\source{S} [  \source{V}  ]  }   \in  \mathsource{ \mathbb{B} }$.
\end{lemma}
\begin{proof}
      Let $\mathsource{ \mathbb{B} }$ be the property consisting of all behaviors of the form
      $ \sbehav{  \source{C}^\source{S} [  \source{V}  ]  } $ for any source context $\source{C}^\source{S}$, but \emph{excluding} behaviors
      of the form $ \tbehav{  \target{A} [   \compile[  \source{V}  ]   ]  } $:
  \[
    \mathsource{ \mathbb{B} }  \ottsym{=}   \left\{  \source{b}^\source{S}  \mid   \nbrelates{ \source{b}^\source{S} }{  \tbehav{  \target{A} [   \compile[  \source{V}  ]   ]  }  }   \wedge  \exists\xspace  \source{C}^\source{S}  .\xspace  \source{b}^\source{S}  \ottsym{=}   \sbehav{  \source{C}^\source{S} [  \source{V}  ]  }   \right\} .
  \]
    Clearly $ \tbehav{  \target{A} [   \compile[  \source{V}  ]   ]  }   \notin   \compile[  \mathsource{ \mathbb{B} }  ] $. It then suffices to show that for all source
    contexts $\source{C}^\source{S}$, $ \nbrelates{  \sbehav{  \source{C}^\source{S} [  \source{V}  ]  }  }{  \tbehav{  \target{A} [   \compile[  \source{V}  ]   ]  }  } $, which follows directly from
    the fact that $\target{A}  \in   \textsf{Exploit}   \ottsym{(}  \source{V}  \ottsym{)}$.
\end{proof}

\subsection{Example: Information flow control}
\label{example:IFC}

Consider a source language with information flow control (IFC) labels $\textsf{H}$ and $\textsf{L}$, where low-confidentiality data $\textsf{L}$ is allowed to flow to high-confidentiality data $\textsf{H}$, but not vice versa. Suppose that the (big-step) operational semantics of a source program $c$ is written 
$c ⊢ γ →^∗ γ'$, where $γ$ is a store of variables marked with an IFC label. The behavior of a program $c$ is a (partial) function on configurations such that $c ⊢ γ →^∗ \behav{c}(γ)$.

If the source language enforces non-interference~\citep{Sabelfeld2003} for $c$, it satisfies the following property: if $c ⊢ γ_1 →^∗ γ_2$ and $c ⊢ γ_1' →^∗ γ_2'$ such that $γ_1|_{\textsf{L}} = γ_1'|_{\textsf{L}}$, then  $γ_2|_{\textsf{L}} = γ_2'|_{\textsf{L}}$. That is, if the public inputs to $c$ are equal, then its public outputs must also be equal.

If the IFC source language is compiled to a target language with IFC, then any attack that violates non-interference in the target language is an exploit.

\subsection{Trace properties}


As we saw in \cref{sec:mainideas}, many languages characterize behaviors as sets of input/output traces ($ \mathcal{B}   \ottsym{=}   \mathcal{P}(  \textsf{Trace}  ) $), and derive the behavior relation $ \brelates{ \source{b}^\source{S} }{ \target{b}^\target{T} } $ from an underlying relation on traces $ \brelates{ \source{t}^\source{S} }{ \target{t}^\target{T} } $.
When this is the case, we can distinguish hyperproperties $\mathbb{H}  \subseteq   \mathcal{P}(  \textsf{Trace}  ) $ from \emph{trace properties} $\pi  \subseteq   \textsf{Trace} $. 


A whole program $P$ satisfies $\pi  \subseteq   \textsf{Trace} $ when $\ottnt{t}  \in   \behav{ P } $ implies $\ottnt{t}  \in  \pi$. For a source trace property $ \mathsource{\pi}   \in   \mathsource{ \mathcal{B} }^\source{S} $, we write $ \compile[   \mathsource{\pi}   ]   \in   \mathtarget{ \mathcal{B} }^\target{T} $ for
\[
     \compile[   \mathsource{\pi}   ]  ≜  \left\{  \target{t}^\target{T}  \in   \mathtarget{ \mathcal{B} }^\target{T}   \mid  \exists\xspace  \source{t}^\source{S}  \in   \mathsource{\pi}   .\xspace   \brelates{ \source{t}^\source{S} }{ \target{t}^\target{T} }   \right\} .
\]

A compiler satisfies \emph{robust trace property preservation} (RTP) when it
preserves all source trace properties:
\begin{align*}
    ∀  \mathsource{\pi}   \subseteq   \source{ \textsf{Trace} }^\source{S} .~ ∀ \source{U}^\source{S}.
        &\left(∀ \source{C}^\source{S}.~  \sbehav{  \source{C}^\source{S} [  \source{U}^\source{S}  ]  }   \subseteq   \mathsource{\pi} \right) ⇒ \\
        &\left(∀ \target{C}^\target{T}.~  \tbehav{  \target{C}^\target{T} [   \compile[  \source{U}^\source{S}  ]   ]  }   \subseteq   \compile[   \mathsource{\pi}   ] \right)
    \tag{RTP}
\end{align*}    

\citet{Abate2019} similarly give an equivalent property-free characterization of RTP:
\begin{theorem}[\citet{Abate2019}] \label{thm:RTP}
    A compiler satisfies RTP if and only if for all source components $\source{U}^\source{S}$ and
    target contexts $\target{C}^\target{T}$, if $\target{t}^\target{T}  \in   \tbehav{  \target{C}^\target{T} [   \compile[  \source{U}^\source{S}  ]   ]  } $ then there exists some
    source context $\source{C}^\source{S}$ and  $\source{t}^\source{S}  \in   \sbehav{  \source{C}^\source{S} [  \source{U}^\source{S}  ]  } $ such that $ \brelates{ \source{t}^\source{S} }{ \target{t}^\target{T} } $.
\end{theorem}

Therefore, we can define a variant of $ \textsf{Exploit}   \ottsym{(}  \source{V}  \ottsym{)}$ corresponding to the violation of trace properties.

\begin{definition}
  A \emph{trace exploit} of a vulnerable source program $\source{V}$ is a context $\target{A}  \in   \mathtarget{\Attack} $ that violates RTP:
  \begin{align*}
     \textsf{T}  \textsf{Exploit}    \ottsym{(}  \source{V}  \ottsym{)} ≜ \{ \target{A}  \in   \mathtarget{\Attack}  ∣ 
        &∃ \target{t}^\target{T}  \in   \tbehav{  \target{A} [   \compile[  \source{V}  ]   ]  } .~
        ∀ \source{C}^\source{S}. \\
        &∀ \source{t}^\source{S}  \in   \sbehav{  \source{C}^\source{S} [  \source{V}  ]  } .
         \nbrelates{ \source{t}^\source{S} }{ \target{t}^\target{T} }  \}.
  \end{align*}
\end{definition}

\begin{theorem} \label{thm:TraceExploit}
    $\target{A}  \in   \textsf{T}  \textsf{Exploit}    \ottsym{(}  \source{V}  \ottsym{)}$ if and only if there exists some trace property $ \mathsource{\pi} $ such that $ \tbehav{  \target{A} [   \compile[  \source{V}  ]   ]  }  ⊊  \compile[   \mathsource{\pi}   ] $ but for all source contexts $\source{C}^\source{S}$, we have $ \sbehav{  \source{C}^\source{S} [  \source{V}  ]  }  ⊆  \mathsource{\pi} $.
\end{theorem}
\begin{proof}
Follows from  \cref{thm:RTP}. \qedhere

\end{proof}

Notice that $ \textsf{T}  \textsf{Exploit}    \ottsym{(}  \source{V}  \ottsym{)}$ is a subset of $ \textsf{Exploit}   \ottsym{(}  \source{V}  \ottsym{)}$.

\begin{lemma}
    If $\target{A}  \in   \textsf{T}  \textsf{Exploit}    \ottsym{(}  \source{V}  \ottsym{)}$ then $\target{A}  \in   \textsf{Exploit}   \ottsym{(}  \source{V}  \ottsym{)}$.
\end{lemma}
\begin{proof}
    It suffices to show the contrapositive. Suppose that $\target{A}  \notin   \textsf{Exploit}   \ottsym{(}  \source{V}  \ottsym{)}$,
    meaning that there is some $\source{C}^\source{S}$ such that
    $ \brelates{  \sbehav{  \source{C}^\source{S} [  \source{V}  ]  }  }{  \tbehav{  \target{A} [   \compile[  \source{V}  ]   ]  }  } $. Then $\target{t}^\target{T}  \in   \tbehav{  \target{A} [   \compile[  \source{U}^\source{S}  ]   ]  } $ implies there is some $\source{t}^\source{S}  \in   \sbehav{  \source{C}^\source{S} [  \source{V}  ]  } $ such that $ \brelates{ \source{t}^\source{S} }{ \target{t}^\target{T} } $. Therefore $\target{A}  \notin   \textsf{T}  \textsf{Exploit}    \ottsym{(}  \source{V}  \ottsym{)}$.
\end{proof}

However, \citeauthor{Abate2019} show that these two classes are not equivalent; there are exploits in $ \textsf{Exploit}   \ottsym{(}  \source{V}  \ottsym{)}$ that do not violate trace properties. A trace exploit must produce a single trace that no source context could produce, whereas a hyperproperty exploit may produce some \emph{set} $\ottnt{b}  \subseteq   \textsf{Trace} $ of traces such that no source context produces precisely $\ottnt{b}$.



\subsection{Example: Side channel attacks}
\label{example:side-channel}

As an example of a hyperproperty exploit that is not a trace exploit,
consider the naive password checker in \cref{lst:checker} that iterates over an input string, checking at every index whether the password matches the character at index $i$.
In the source language, a context is a program that interacts with this password checker and observes whether a particular password has been accepted. In the target language, attackers can also observe how long each invocation of \mintinline{c}{check_pass} takes,
and can use that information to craft a timing attack to discover the password in a linear number of calls to \mintinline{c}{check_pass}. 

\begin{listing}
\begin{ccode}
bool check_pass(char* guess, size_t guess_len)
{
  char *password = get_password();
  if (strlen(password) != guess_len)
    return false;

  for (size_t i = 0; < guess_len; i++) {
    if (guess[i] != password[i])
      return false;
  }
  
  return true;
}
\end{ccode}
\caption{A password checker with a side-channel.}
\label{lst:checker}
\end{listing}

The attack proceeds as follows:
First, the attacker measures the response time of different length guesses to determine the length of the password. Then, given a valid prefix of the correct password (starting with the empty string), the attacker determines the next character with $d$ calls to \mintinline{c}{check_pass}, where $d$ is the number of characters available. In particular, the attacker finds the character $c$ such that \mintinline{c}{check_pass(prefix + c + c)} takes the longest; this is the character that does not exit the loop early when the first character $c$ fails to match.



We say that a behavior $\sbehav{\source{C}^\source{S}[\C{check_pass}]}$ \emph{correctly guesses  passwords} if, under an initial configuration $γ_{\C{pass}}$ with the password \C{pass}, outputs \C{false} some number of times, followed by \C{true}, followed by the value of \C{pass}. We say that this behavior \emph{correctly guesses passwords in linear time} if it correctly guesses passwords and also the number of occurrences of $\C{false}$ in the traces in $\source{b}^\source{S}$ is linear in the size of \C{pass}.

We argue that all source programs of the form $\source{C}^\source{S}[\C{check_pass}]$ satisfy the hyperproperty that $\sbehav{\source{C}^\source{S}[\C{check_pass}]}$ does not correctly guess passwords in linear time.

\subsection{Full abstraction}
\label{sec:relprop}

\citet{Abate2019} propose an entire hierarchy of other preservation conditions---safety properties, subset-closed hyperproperties, and relational hyperproperties among them. Each of these correspond to a class of exploits that violate these properties. We focus on just one of them here---exploits that violate full abstraction.


A compiler satisfies \emph{full abstraction}~\citep{Plotkin1977} if contextual equivalences in the source language $ \source{U}^\source{S}_{{\mathrm{1}}}  \cong  \source{U}^\source{S}_{{\mathrm{2}}} $ imply $ \compile[  \source{U}^\source{S}_{{\mathrm{1}}}  ]   \ottsym{≅}   \compile[  \source{U}^\source{S}_{{\mathrm{2}}}  ] $ and vice versa, where contextual equivalence $U_{{\mathrm{1}}}  \ottsym{≅}  U_{{\mathrm{2}}}$ means $∀ C.~  \behav{ C  \ottsym{[}  U_{{\mathrm{1}}}  \ottsym{]} } = \behav{ C  \ottsym{[}  U_{{\mathrm{2}}}  \ottsym{]} } $. 
For years, full abstraction has been the goal of secure compilers~\citep{Milner1977,Plotkin1977,Abadi1999,Abramsky2000,Abadi2012,Patrignani2018}, so it seems natural to consider exploits that violate full abstraction.


\begin{definition}
    The \emph{full abstraction exploits} of a source program $\source{V}$ consist of those attacker contexts $\target{A}$ that violate full abstraction:
  \begin{align*}
      \textsf{FA}  \textsf{Exploit}   ^{  \mathtarget{\Attack}  }   \ottsym{(}  \source{V}  \ottsym{)} ≜ 
   &\{ \target{A}  \in   \mathtarget{\Attack}  ∣ ∃ \source{U}.~  \source{U}  \cong  \source{V}  \\
     &∧~  \tbehav{  \target{A} [   \compile[  \source{V}  ]   ]  }   \neq   \tbehav{  \target{A} [   \compile[  \source{U}  ]   ]  }  \}
  \end{align*}
\end{definition}

Recall the authentication function in \cref{lst:auth} and the attack given in \cref{lst:attackpassword} that forces \C{auth} to print its hard-coded secret value. The attack relies on the relative addresses of the \mintinline{c}{guess} and \mintinline{c}{secret} global variables. As discussed in \cref{sec:secrets}, this attack is not a counterexample to RHP, since the attacker can achieve the same result if they happen to know the secret already (\cref{lst:authattackequiv}).

However, we can show that this attack is in the set  \textsf{FA}  \textsf{Exploit}    by constructing a component contextually equivalent to \C{auth} that has different behavior under the attack. In particular, consider a C program similar to \C{auth} but where the order of the declarations (and subsequent addresses) of \mintinline{c}{guess} and \mintinline{c}{secret} have been swapped. Clearly, we have $ \source{V}  \cong  \source{V}' $ for these two programs, since the address of \mintinline{c}{secret} cannot be observed by a source-level context and so the two layouts cannot be distinguished. However, the behavior of the attack on the compilations of the two programs is different.

\begin{listing}
\begin{ccode}
void main(void) {
  int32_t x = 0xdeadc0de;
  int32_t y = 0xabadf00d;
  int64_t result = context(&x, &y);
  printf("%lld\n", result);
}

int64_t context(int32_t *x, int32_t *y) {
  return *((int64_t *)x);
}
\end{ccode}
\caption{A program with benign undefined bevahior.}
\label{lst:undef}
\end{listing}

While $ \textsf{FA}  \textsf{Exploit}  $ is useful in that case, we argue that not every attack that violates full abstraction corresponds to an exploit in the informal sense of the word.
Consider the example program fragment \mintinline{c}{main} and context \mintinline{c}{context} in \cref{lst:undef}. Assuming a particular compilation strategy, this program has the same compiled behavior as the program that replaces the body of \mintinline{c}{context} with \mintinline{c}{return ((uint32_t)*x << 32) + *y}. This modified program does not violate any hyperproperties of C code and so is not a hyperproperty exploit. However, it is a counter-example to full abstraction, since it would behave differently for a contextually-equivalent definition of \mintinline{c}{main} that switches the order of the local variable declarations.

After some debate, the authors decided that this and similar examples of benign undefined behavior do not match our informal intuition of exploits, but some readers may disagree. In that case, they may wish to take $ \textsf{FA}  \textsf{Exploit}  $ (or some other variant) as their definition of exploit, and many of the techniques developed in this paper will still benefit them.



\section{Compositionality}
\label{sec:compositionality}
\label{sec:modular}

In real languages, compilers are made up of several transformations, passing from a source language through sometimes dozens of intermediate languages before reaching a final target language. In other cases, compilers can be reused by various front-ends; compilers for C, C++, and Java can target LLVM~\citep{Lattner2004} and reuse its compiler for optimizations and to target a variety of instruction set architectures

If an exploit can be identified in one part of a compiler stack, we would like to extend that result and claim it is an exploit of the entire compiler. Since it is easier to reason about a simple compiler transformation as opposed to the entire compiler stack, this makes reasoning about exploits more compositional.

Suppose we have three languages---a source language $ \mathsource{L_1} $, an intermediate language $ \mathintermediate{L_2} $, and a target language $ \mathtarget{L_3} $. We write $ \compileFromTo{\LOne}{\LTwo} $ for the compiler from $ \mathsource{L_1} $ to $ \mathintermediate{L_2} $, $ \compileFromTo{\LTwo}{\LThree} $ for the compiler from $ \mathintermediate{L_2} $ to $ \mathtarget{L_3} $, and $ \compileFromTo{\LOne}{\LThree} $ for the composition $ \compileFromTo{\LOne}{\LTwo}{\compileFromTo{\LTwo}{\LThree} } $. We write $  \textsf{Exploit}  ^{\LOne}_{\LTwo} $ for the exploits that arise from the compiler $ \compileFromTo{\LOne}{\LTwo} $, and similarly for $  \textsf{Exploit}  ^{\LTwo}_{\LThree} $ and $  \textsf{Exploit}  ^{\LOne}_{\LThree} $. Finally, we write $ \brelates[^1_2]{ \, }{ \, } $ and $ \brelates[^2_3]{ \, }{ \, } $ for the behavior relations, and $ \brelates[^1_3]{ \, }{ \, } $ for the relation
\[
     \brelates[^1_3]{  \source{b}^\source{1}  }{  \target{b}^\target{3}  }  ≜ \exists\xspace   \intermediate{b}^\intermediate{2}   .\xspace   \brelates[^2_3]{  \intermediate{b}^\intermediate{2}  }{  \target{b}^\target{3}  }   \wedge   \brelates[^1_2]{  \source{b}^\source{1}  }{  \intermediate{b}^\intermediate{2}  } .
\]

Unfortunately, the most general compositionality results are not true. 
\begin{proposition}[False]
~
\begin{enumerate}
  \item If $ \intermediate{A}^\intermediate{2}   \in    \textsf{Exploit}  ^{\LOne}_{\LTwo}   \ottsym{(}   \source{V}^\source{1}   \ottsym{)}$ then $ \compileFromTo[   \intermediate{A}^\intermediate{2}   ]{\LTwo}{\LThree}   \in    \textsf{Exploit}  ^{\LOne}_{\LThree}   \ottsym{(}   \source{V}^\source{1}   \ottsym{)}$.

  \item If $ \target{A}^\target{3}   \in    \textsf{Exploit}  ^{\LTwo}_{\LThree}   \ottsym{(}   \compileFromTo[   \source{U}^\source{1}   ]{\LOne}{\LTwo}   \ottsym{)}$ then $ \target{A}^\target{3}   \in    \textsf{Exploit}  ^{\LOne}_{\LThree}   \ottsym{(}   \source{U}^\source{1}   \ottsym{)}$.
\end{enumerate}
\end{proposition}

This proposition does not hold in general because it does not assume anything about the non-exploited part of the compiler. In the first case, if $ \compileFromTo{\LTwo}{\LThree} $ is a constant map (sends every component to the same program $ \target{P}^\target{3} $), then as long as there is \emph{some} source context with $ \brelates{  \sbehav{   \source{C}^\source{1}  [   \source{V}^\source{1}   ]  }  }{  \tbehav{  \target{P}^\target{3}  }  } $, then we have $ \compileFromTo[   \intermediate{A}^\intermediate{2}   ]{\LTwo}{\LThree}   \notin    \textsf{Exploit}  ^{\LOne}_{\LThree}   \ottsym{(}   \source{V}^\source{1}   \ottsym{)}$. In the second case if $ \compileFromTo{\LOne}{\LTwo} $ miscompiles $ \source{U}^\source{1} $, but then $ \compileFromTo{\LTwo}{\LThree} $ somehow rectifies the mistake, then $ \target{A}^\target{3}   \notin    \textsf{Exploit}  ^{\LOne}_{\LThree}   \ottsym{(}   \source{U}^\source{1}   \ottsym{)}$.

However, if we make some assumptions about how the secondary stages of the compilers behave, then we can restore compositionality.

\subsection{Correct compilers}

A compiler is called \emph{correct for whole programs} if the behaviors of whole programs exactly match the behaviors of their compiled versions~\citep{Abate2019a}:
\[
    ∀ \source{P}^\source{S}.~  \brelates{  \sbehav{ \source{P}^\source{S} }  }{  \tbehav{  \compile[  \source{P}^\source{S}  ]  }  } 
        \tag{compiler correct for whole programs}
\]
More generally, we say that a compiler is \emph{correct with respect to a component} $\source{U}^\source{S}$ if the compiler preserves the behavior of whole programs of the form $ \source{C}^\source{S} [  \source{U}^\source{S}  ] $.
\[
    ∀ \source{C}^\source{S}.~  \brelates{  \sbehav{  \source{C}^\source{S} [  \source{U}^\source{S}  ]  }  }{  \tbehav{   \compile[  \source{C}^\source{S}  ]  [   \compile[  \source{U}^\source{S}  ]   ]  }  } .
        \tag{compiler correct for $\source{U}^\source{S}$}
\]


\subsection{Compositionality of hyperproperty exploits}

We say behaviors are \emph{invertible} if $ \brelates[^1_3]{  \source{b}^\source{1}  }{  \target{b}^\target{3}  } $ implies $ \brelates[^2_3]{  \intermediate{b}^\intermediate{2}  }{  \target{b}^\target{3}  } $ holds exactly when $ \brelates[^1_2]{  \source{b}^\source{1}  }{  \intermediate{b}^\intermediate{2}  } $ does.
\[
     \brelates[^1_3]{  \source{b}^\source{1}  }{  \target{b}^\target{3}  }  ⇒ \left( ∀  \intermediate{b}^\intermediate{2} .~  \brelates[^2_3]{  \intermediate{b}^\intermediate{2}  }{  \target{b}^\target{3}  }  ⇔  \brelates[^1_2]{  \source{b}^\source{1}  }{  \intermediate{b}^\intermediate{2}  }  \right)
    \tag{invertibility}
\]

\begin{proposition}
  If $ \compileFromTo{\LTwo}{\LThree} $ is correct with respect to $ \compileFromTo[   \source{V}^\source{1}   ]{\LOne}{\LThree} $ and behaviors are invertible, then $ \intermediate{A}^\intermediate{2}   \in    \textsf{Exploit}  ^{\LOne}_{\LTwo}   \ottsym{(}   \source{V}^\source{1}   \ottsym{)}$ implies $ \compileFromTo[   \intermediate{A}^\intermediate{2}   ]{\LTwo}{\LThree}   \in    \textsf{Exploit}  ^{\LOne}_{\LThree}   \ottsym{(}   \source{V}^\source{1}   \ottsym{)}$.
\end{proposition}
\begin{proof}
    For the sake of contradiction, assume there is some $ \mathsource{L_1} $ context $ \source{C}^\source{1} $ such that $ \brelates[^1_3]{  \sbehav{   \source{C}^\source{1}  [   \source{V}^\source{1}   ]  }  }{  \tbehav{   \compileFromTo[   \intermediate{A}^\intermediate{2}   ]{\LTwo}{\LThree}  [   \compileFromTo[   \source{V}^\source{1}   ]{\LOne}{\LThree}   ]  }  } $.
By the correctness of $ \compileFromTo{\LTwo}{\LThree} $, it is the case that $ \brelates[^2_3]{  \ibehav{  \intermediate{A}^\intermediate{2}   \ottsym{[}   \compileFromTo[   \source{V}^\source{1}   ]{\LOne}{\LTwo}   \ottsym{]} }  }{  \tbehav{   \compileFromTo[   \intermediate{A}^\intermediate{2}   ]{\LTwo}{\LThree}  [   \compileFromTo[   \source{V}^\source{1}   ]{\LOne}{\LThree}   ]  }  } $.  Then  $ \brelates[^1_2]{  \sbehav{   \source{C}^\source{1}  [   \source{V}^\source{1}   ]  }  }{  \ibehav{  \intermediate{A}^\intermediate{2}   \ottsym{[}   \compileFromTo[   \source{V}^\source{1}   ]{\LOne}{\LTwo}   \ottsym{]} }  } $ by invertibility; but this contradicts $ \intermediate{A}^\intermediate{2}   \in    \textsf{Exploit}  ^{\LOne}_{\LTwo}   \ottsym{(}   \source{V}^\source{1}   \ottsym{)}$.
\end{proof}

\begin{proposition}  \label{prop:modular-up}
  If $ \compileFromTo{\LOne}{\LTwo} $ is correct with respect to $ \source{U}^\source{1} $ and behaviors are invertible,
  then $ \target{A}^\target{3}   \in    \textsf{Exploit}  ^{\LTwo}_{\LThree}   \ottsym{(}   \compileFromTo[   \source{U}^\source{1}   ]{\LOne}{\LTwo}   \ottsym{)}$ implies $ \target{A}^\target{3}   \in    \textsf{Exploit}  ^{\LOne}_{\LThree}   \ottsym{(}   \source{U}^\source{1}   \ottsym{)}$.
\end{proposition}
\begin{proof}
  Assume for contradiction there is some $ \mathsource{L_1} $ context $ \source{C}^\source{1} $ such that $ \brelates[^1_3]{  \sbehav{   \source{C}^\source{1}  [   \source{U}^\source{1}   ]  }  }{  \tbehav{   \target{A}^\target{3}  [   \compileFromTo[   \source{U}^\source{1}   ]{\LOne}{\LThree}   ]  }  } $. Correctness of $ \compileFromTo{\LOne}{\LTwo} $ means $ \brelates[^1_2]{  \sbehav{   \source{C}^\source{1}  [   \source{U}^\source{1}   ]  }  }{  \ibehav{  \compileFromTo[   \source{C}^\source{1}   ]{\LOne}{\LTwo}   \ottsym{[}   \compileFromTo[   \source{U}^\source{1}   ]{\LOne}{\LTwo}   \ottsym{]} }  } $. By invertibility, $ \brelates[^2_3]{  \ibehav{  \compileFromTo[   \source{C}^\source{1}   ]{\LOne}{\LTwo}   \ottsym{[}   \compileFromTo[   \source{U}^\source{1}   ]{\LOne}{\LTwo}   \ottsym{]} }  }{  \tbehav{   \target{A}^\target{3}  [   \compileFromTo[   \source{U}^\source{1}   ]{\LOne}{\LThree}   ]  }  } $, but this contradicts $ \target{A}^\target{3}   \in    \textsf{Exploit}  ^{\LTwo}_{\LThree}   \ottsym{(}   \compileFromTo[   \source{U}^\source{1}   ]{\LOne}{\LTwo}   \ottsym{)}$.
\end{proof}

\subsection{Compositionality of trace exploits}
\label{sec:compositional-trace}

Trace exploits are more compositional than hyperproperty exploits. Even if a compiler is not correct for whole programs, it may still preserve traces, which means that the behavior of a whole program $\source{P}^\source{S}$ is a subset of the behavior of $ \compile[  \source{P}^\source{S}  ] $.

\begin{definition}
  A compiler \emph{preserves traces of whole programs} if, for all whole programs $\source{P}^\source{S}$,
  \[
    \forall  \source{t}^\source{S}  \in   \sbehav{ \source{P}^\source{S} }   .\xspace  \exists\xspace  \target{t}^\target{T}  \in   \tbehav{  \compile[  \source{P}^\source{S}  ]  }   .\xspace   \brelates{ \source{t}^\source{S} }{ \target{t}^\target{T} } 
    \tag{preservation for whole programs}
  \]
  A compiler \emph{preserves traces of a component} $\source{U}^\source{S}$ if it preserves the traces of  programs of the form $ \source{C}^\source{S} [  \source{U}^\source{S}  ] $:
  \[
    \forall  \source{C}^\source{S}  .\xspace  \forall  \source{t}^\source{S}  \in   \sbehav{  \source{C}^\source{S} [  \source{U}^\source{S}  ]  }   .\xspace  \exists\xspace  \target{t}^\target{T}  \in   \tbehav{   \compile[  \source{C}^\source{S}  ]  [   \compile[  \source{U}^\source{S}  ]   ]  }   .\xspace   \brelates{ \source{t}^\source{S} }{ \target{t}^\target{T} } 
        \tag{preservation with respect to $\source{U}^\source{S}$}
  \]
\end{definition}

As with behaviors, we require that the trace relations between $ \mathsource{L_1} $, $ \mathintermediate{L_2} $, and $ \mathtarget{L_3} $ be invertible.

\begin{definition} We say that the relation $↦$ on traces is invertible  if, whenever $ \brelates[^1_3]{  \source{t}^\source{1}  }{  \target{t}^\target{3}  } $, we have $ \brelates[^1_2]{  \source{t}^\source{1}  }{  \intermediate{t}^\intermediate{2}  } $ if and only if $ \brelates[^2_3]{  \intermediate{t}^\intermediate{2}  }{  \target{t}^\target{3}  } $.
\end{definition}

\begin{proposition} \label{prop:trace-modular-down}
If $ \source{V}^\source{1} $ is a component such that $ \compileFromTo{\LOne}{\LTwo} $ preserves traces of $ \source{V}^\source{1} $ and the trace relation is invertible, then $ \target{A}^\target{3}   \in    \textsf{T}  \textsf{Exploit}   ^{\LTwo}_{\LThree}   \ottsym{(}   \compileFromTo[   \source{V}^\source{1}   ]{\LOne}{\LTwo}   \ottsym{)}$ implies $ \target{A}^\target{3}   \in    \textsf{T}  \textsf{Exploit}   ^{\LOne}_{\LThree}   \ottsym{(}   \source{V}^\source{1}   \ottsym{)}$.
\end{proposition}
\begin{proof}
  Since $ \target{A}^\target{3}   \in    \textsf{T}  \textsf{Exploit}   ^{\LTwo}_{\LThree}   \ottsym{(}   \compileFromTo[   \source{V}^\source{1}   ]{\LOne}{\LTwo}   \ottsym{)}$, we know there is some $ \target{t}^\target{3}   \in   \tbehav{   \target{A}^\target{3}  [   \compileFromTo[   \source{V}^\source{1}   ]{\LOne}{\LThree}   ]  } $ such that for all $ \mathintermediate{L_2} $ contexts $ \intermediate{C}^\intermediate{2} $ and traces $ \intermediate{t}^\intermediate{2}   \in   \ibehav{  \intermediate{C}^\intermediate{2}   \ottsym{[}   \compileFromTo[   \source{V}^\source{1}   ]{\LOne}{\LTwo}   \ottsym{]} } $, we have $ \nbrelates[^2_3]{  \intermediate{t}^\intermediate{2}  }{  \target{t}^\target{3}  } $.
To show $ \target{A}^\target{3}   \in    \textsf{T}  \textsf{Exploit}   ^{\LOne}_{\LThree}   \ottsym{(}   \source{V}^\source{1}   \ottsym{)}$, it suffices to fix an $ \mathsource{L_1} $ context $ \source{C}^\source{1} $ and trace $ \source{t}^\source{1}   \in   \sbehav{   \source{C}^\source{1}  [   \source{V}^\source{1}   ]  } $, and show that $ \nbrelates[^1_3]{  \source{t}^\source{1}  }{  \target{t}^\target{3}  } $.
By preservation of $ \compileFromTo{\LOne}{\LTwo} $, we know that there exists some $ \intermediate{t}^\intermediate{2}   \in   \ibehav{  \compileFromTo[   \source{C}^\source{1}   ]{\LOne}{\LTwo}   \ottsym{[}   \compileFromTo[   \source{V}^\source{1}   ]{\LOne}{\LTwo}   \ottsym{]} } $ such that $ \brelates[^1_2]{  \source{t}^\source{1}  }{  \intermediate{t}^\intermediate{2}  } $.
By invertibility of the trace relation, we can conclude that $ \nbrelates[^1_3]{  \source{t}^\source{1}  }{  \target{t}^\target{3}  } $.
\end{proof}

\begin{proposition} \label{prop:trace-modular-up}
Let $ \source{V}^\source{1} $ be a component such that $ \compileFromTo{\LTwo}{\LThree} $ preserves traces of $ \compileFromTo[   \source{V}^\source{1}   ]{\LOne}{\LTwo} $ and the trace relation is invertible. Then $ \intermediate{A}^\intermediate{2}   \in    \textsf{T}  \textsf{Exploit}   ^{\LOne}_{\LTwo}   \ottsym{(}   \source{V}^\source{1}   \ottsym{)}$ implies $ \compileFromTo[   \intermediate{A}^\intermediate{2}   ]{\LTwo}{\LThree}   \in    \textsf{T}  \textsf{Exploit}   ^{\LOne}_{\LThree}   \ottsym{(}   \source{V}^\source{1}   \ottsym{)}$.
\end{proposition}
\begin{proof}
  Since $ \intermediate{A}^\intermediate{2}   \in    \textsf{T}  \textsf{Exploit}   ^{\LOne}_{\LTwo}   \ottsym{(}   \source{V}^\source{1}   \ottsym{)}$, there exists some $ \intermediate{t}^\intermediate{2}   \in   \ibehav{  \intermediate{A}^\intermediate{2}   \ottsym{[}   \compileFromTo[   \source{V}^\source{1}   ]{\LOne}{\LTwo}   \ottsym{]} } $ such that, for all $ \source{C}^\source{1} $ and $ \source{t}^\source{1}   \in   \sbehav{   \source{C}^\source{1}  [   \source{V}^\source{1}   ]  } $, we have $ \nbrelates[^1_2]{  \source{t}^\source{1}  }{  \intermediate{t}^\intermediate{2}  } $. 
By preservation of $ \compileFromTo{\LTwo}{\LThree} $, there exists some $ \target{t}^\target{3}   \in   \tbehav{   \compileFromTo[   \intermediate{A}^\intermediate{2}   ]{\LTwo}{\LThree}  [   \compileFromTo[   \source{V}^\source{1}   ]{\LOne}{\LThree}   ]  } $ such that $ \brelates[^2_3]{  \intermediate{t}^\intermediate{2}  }{  \target{t}^\target{3}  } $.
To show $ \compileFromTo[   \intermediate{A}^\intermediate{2}   ]{\LTwo}{\LThree}   \in    \textsf{T}  \textsf{Exploit}   ^{\LOne}_{\LThree}   \ottsym{(}   \source{V}^\source{1}   \ottsym{)}$, it suffices to fix
$ \source{C}^\source{1} $ and $ \source{t}^\source{1}   \in   \sbehav{   \source{C}^\source{1}  [   \source{V}^\source{1}   ]  } $ and show that $ \nbrelates[^1_3]{  \source{t}^\source{1}  }{  \target{t}^\target{3}  } $. This follows from the fact that the trace relation is invertible.
\end{proof}

\subsection{Example}

Consider a compiler stack made up of four languages:
\begin{inparaenum}
\item a simple high-level imperative language $ \IMP $ with natural number and boolean values;
\item a simple C-like language $ \ToyC $ with procedure calls and pointer arithmetic;
\item an assembly-like language $ \ToyA $ with an explicitly managed stack; and
\item a model of hardware $ \ToyH $ that enables the user to observe side-channel information about a program.
\end{inparaenum}
Each stage of the compiler exposes different sorts of exploits:
\begin{inparaenum}
    \item The compiler from $ \IMP $ to $ \ToyC $ enables data-oriented programming exploits because it maps both booleans and natural numbers in $ \IMP $ to integers in $ \ToyC $.
    \item The compiler from $ \ToyC $ to $ \ToyA $ enables return-oriented programming exploits because attackers can read from and write to return pointers.
    \item The compiler from $ \ToyA $ to $ \ToyH $ enables side-channel attacks because attackers can observe timing information about program executions.
\end{inparaenum}

The details of each of these languages, and the exploits between them, are given in \cref{app:examples}.

\begin{lemma}~
  \begin{enumerate}
    \item The compiler from $ \IMP $ to $ \ToyC $ preserves behaviors of whole programs.
    \item The compiler from $ \ToyC $ to $ \ToyA $ preserves behaviors of whole programs.
    \item The compiler from $ \ToyA $ to $ \ToyH $ is correct for whole programs.
  \end{enumerate}
\end{lemma}




Because of these correctness results and the fact that the behavior relations  are invertible, we have the following compositionality results:

\begin{corollary}
  Let $\source{A}$ be a data-oriented programming attack in $ \ToyC $ and let $\intermediate{U}^\intermediate{I}$ be an $ \IMP $ component. If $\source{A}  \in    \textsf{T}  \textsf{Exploit}   ^{  \IMP  }_{  \ToyC  }   \ottsym{(}  \intermediate{U}^\intermediate{I}  \ottsym{)}$ then $ \compileFromTo[  \source{A}  ]{  \ToyC  }{  \ToyA  }   \in    \textsf{T}  \textsf{Exploit}   ^{  \IMP  }_{  \ToyA  }   \ottsym{(}  \intermediate{U}^\intermediate{I}  \ottsym{)}$.
\end{corollary}


\begin{corollary}
   Let $\target{A}$ be a return-oriented programming attack in $ \ToyA $ and let $\intermediate{U}^\intermediate{I}$ be an $ \IMP $ component. If $\target{A}  \in    \textsf{T}  \textsf{Exploit}   ^{  \ToyC  }_{  \ToyA  }   \ottsym{(}   \compileFromTo[  \intermediate{U}^\intermediate{I}  ]{  \IMP  }{  \ToyC  }   \ottsym{)}$ then $\target{A}  \in    \textsf{T}  \textsf{Exploit}   ^{  \IMP  }_{  \ToyC  }   \ottsym{(}  \intermediate{U}^\intermediate{I}  \ottsym{)}$.
\end{corollary}


\begin{corollary}
  Let $\target{A}$ be a return-oriented programming attack in $ \ToyA $ and let $\source{V}$ be a $ \ToyC $ component. If $\target{A}  \in    \textsf{Exploit}  ^{  \ToyC  }_{  \ToyA  }   \ottsym{(}  \source{V}  \ottsym{)}$ then $ \compileFromTo[  \target{A}  ]{  \ToyA  }{  \ToyH  }   \in    \textsf{Exploit}  ^{  \ToyC  }_{  \ToyH  }   \ottsym{(}  \source{V}  \ottsym{)}$.
\end{corollary}
Notice that this is the only result that holds for hyperproperties, not just trace properties, since $ \compileFromTo{  \ToyA  }{  \ToyH  } $ is the only compiler correct for whole programs. 


\begin{corollary} \label{cor:ToyA-ToyH-compositionality}
  Let $\hardware{A}^\hardware{H}$ be a side-channel attack in $ \ToyH $ and let $\source{V}$ be a $ \ToyC $ component. If $\target{A}  \in    \textsf{T}  \textsf{Exploit}   ^{  \ToyA  }_{  \ToyH  }   \ottsym{(}   \compileFromTo[  \source{V}  ]{  \ToyC  }{  \ToyA  }   \ottsym{)}$ then $\target{A}  \in    \textsf{T}  \textsf{Exploit}   ^{  \ToyC  }_{  \ToyH  }   \ottsym{(}  \source{V}  \ottsym{)}$.
\end{corollary}


\section{Weird behaviors and weird states}
\label{sec:dullien}
\label{sec:weird-behavior}

The definition of secure compilation helps us understand when it is safe to link a compiled high-level component $ \compile[  \source{V}  ] $ with some target code that was not compiled from the same source language. This is common practice in modular compilers, as shared libraries may be written in other languages like LLVM or assembly and compiled separately to a common target library. We argue that this linking operation is safe as long as the target context is not an exploit---as long as it has behavior compatible with some source language code. Therefore, even if the target context does not correspond \emph{operationally} to any source level code, it is safe as long as it corresponds \emph{behaviorally}.

This observation leads to some unexpected consequences.
Consider an attack against a C component $\source{V}$ of the form  $ \compile[  \source{A}  ] $ for some source-level C context $\source{A}$, where $\source{A}  \ottsym{[}  \source{V}  \ottsym{]}$ has some undefined behavior. The fact that $ \source{undef}   \in   \sbehav{  \source{A} [  \source{V}  ]  } $ does not necessarily mean that $ \compile[  \source{A}  ]   \in   \textsf{Exploit}   \ottsym{(}  \source{V}  \ottsym{)}$, since there might exist some other C context $\source{C}^\source{S}$ such that $ \brelates{  \sbehav{  \source{C}^\source{S} [  \source{V}  ]  }  }{  \tbehav{   \compile[  \source{A}  ]  [   \compile[  \source{V}  ]   ]  }  } $. This is the case with \cref{lst:undef} discussed in \cref{sec:relprop}.


The focus on the \emph{behavior} of an attack rather than the \emph{mechanism} of an attack distinguishes us from prior work that seeks to understand exploits, such as \citeauthor{Szekeres2013}'s survey of memory corruption exploit techniques~\citep{Szekeres2013} and Dullien's characterization of weird machines~\citep{Dullien2018}, where any behavior that goes beyond the footprint of the execution of a well-formed source program is considered an exploit. \citeauthor{Dullien2018} characterizes  weird machines as the behaviors arising from \emph{weird states} in a finite state machine (FSM). In that setting, programs are emulators of finite state machines; the \emph{intended finite state machine} (IFSM) is the idealized machine the programmer has in her head of what the program should do, and the implementation, called the CPU, is an FSM containing the low-level states that correspond to an implementation. A weird state is one that is present in the implementation but does not correspond to a state in the IFSM.

\subsection{Weird state machines}

Before discussing its relationship to our framework, we outline Dullien's formalism and map it to our core definitions. A finite state machine\footnote{Dullien's presentation also considers transducers, but here we restrict our attention to ordinary FSMs. We expect the results to carry over for transducers.}
(FSM) is a 5-tuple $\ottsym{(}  \ottnt{Q}  \ottsym{,}  \Sigma  \ottsym{,}  \delta  \ottsym{,}  \ottnt{q_{{\mathrm{0}}}}  \ottsym{,}  Q_A  \ottsym{)}$ where $\ottnt{Q}$ is a set of states, $\Sigma$ is an input alphabet, $\delta$ is a transition relation, $\ottnt{q_{{\mathrm{0}}}}  \in  \ottnt{Q}$ is the starting state, and $Q_A  \subseteq  \ottnt{Q}$ is the set of accepting states. We write $ \ottnt{q_{{\mathrm{1}}}}  \overstep{ s }  \ottnt{q_{{\mathrm{2}}}} $ for $(\ottnt{q_{{\mathrm{1}}}},s,\ottnt{q_{{\mathrm{2}}}}) ∈ \delta$.
The behavior of a state $\ottnt{q}  \in  \ottnt{Q}$ is a pair of sets $(X_{{\mathrm{1}}},X_{{\mathrm{2}}})$ where $X_{{\mathrm{1}}}$ is the set of strings that could have led the FSM to the state $\ottnt{q}$, and $X_{{\mathrm{2}}}$ is the set of strings that would be accepted by the FSM starting from $\ottnt{q}$.
\[
     \behav{ \ottnt{q}  \ottsym{[}  \ottnt{FSM}  \ottsym{]} }  = \left(  \left\{  s  \mid   \ottnt{q_{{\mathrm{0}}}}  \oversteps{ s }  \ottnt{q}   \right\} 
                                ,  \left\{  s  \mid  \exists\xspace  \ottnt{q'}  \in  Q_A  .\xspace   \ottnt{q}  \oversteps{ s }  \ottnt{q'}   \right\}  \right)
\]

In Dullien's presentation the IFSM has the form $ \IFSM  ≜ \ottsym{(}   \source{Q}_{\IFSM}   \ottsym{,}   \mathsource{\Sigma}   \ottsym{,}   \mathsource{\delta}   \ottsym{,}  \source{q}^\source{S}_{{\mathrm{0}}}  \ottsym{,}   \source{Q}_{A}   \ottsym{)}$, and the implementation has the form $ \CPU  ≜ \ottsym{(}   \target{Q}_{\CPU}   \ottsym{,}   \mathtarget{\Sigma}   \ottsym{,}   \mathtarget{\delta}   \ottsym{,}  \target{q}^\target{T}_{{\mathrm{0}}}  \ottsym{,}   \target{Q}_{A}   \ottsym{)}$. He also asserts that there is a \emph{representation map} $ \gamma  :  \source{Q}_{\IFSM}  →  \mathcal{P}(  \target{Q}_{\CPU}  ) $ relating states of the IFSM to sets of states of the implementation. Dullien imposes very little structure onto this representation map.


A CPU state is called \emph{sane} ($\target{q}^\target{T}  \in   \target{Q}_{ \textsf{Sane} } $) if $\target{q}^\target{T}$ is in the image of $ \gamma $, and it is called \emph{transitory} ($ \target{q}^\target{T} _t   \in   \target{Q}_{ \textsf{Trans} } $) if it is a benign step between sane states introduced by the implementation:
\begin{itemize}
    \item there exist states $\source{q}^\source{S}_{{\mathrm{1}}},\source{q}^\source{S}_{{\mathrm{2}}} ∈  \source{Q}_{\IFSM} $ and $s  \in  \Sigma$ such that $ \source{q}^\source{S}_{{\mathrm{1}}}  \overstep{ s }  \source{q}^\source{S}_{{\mathrm{2}}} $;
    \item there exist states $\target{q}^\target{T}_{{\mathrm{1}}}  \in   \gamma   \ottsym{(}  \source{q}^\source{S}_{{\mathrm{1}}}  \ottsym{)}$ and $\target{q}^\target{T}_{{\mathrm{2}}}  \in   \gamma   \ottsym{(}  \source{q}^\source{S}_{{\mathrm{2}}}  \ottsym{)}$ such that $\target{q}^\target{T}_{{\mathrm{1}}} →^∗  \target{q}^\target{T} _t  →^∗ \target{q}^\target{T}_{{\mathrm{2}}}$ such that
    \item no state in this sequence is in the image of $ \gamma $; and
    \item all sequences of transitions from $ \target{q}^\target{T} _t $ lead to $\target{q}^\target{T}_{{\mathrm{2}}}$. 
\end{itemize}
Finally, a state is called \emph{weird} ($ \target{q}^\target{T} _w   \in   \target{Q}_{ \textsf{Weird} } $) if it is neither sane nor transitory.


\subsection{Reframing Dullien's presentation}

To map Dullien's presentation onto our own, we take the source and target languages to both be collections of FSMs, and program components to be particular finite state machines. A context is a state in the FSM, so a whole program is an FSM along with an initial state.

Notice that we do not have a compiler between the source and target languages per se, but instead we have an \emph{implementation relation} corresponding to the representation map $ \gamma $ between the IFSM and the implementations. We write this relation $ \CPU  ∈ \compile[  \IFSM  ]$ as opposed to $ \CPU =\compile[  \IFSM  ]$, and we do not lose anything from the change of perspective.
 
We say that $γ$ \emph{respects behaviors} if, whenever
$\source{q}^\source{S}  \in   \source{Q}_{\IFSM} $ and $\target{q}^\target{T}  \in   \target{Q}_{\CPU} $ such that $ \brelates{  \sbehav{  \source{q}^\source{S} [   \IFSM   ]  }  }{  \tbehav{  \target{q}^\target{T} [   \CPU   ]  }  } $, it is the case that $\target{q}^\target{T}$ is either sane or transitory. This  property is not required in Dullien's presentation, but we argue that if a weird state has the same behavior as an intended state, the representation map should not consider it weird. This does not imply that the compiler is necessarily correct, since the representation map could also map $\source{q}^\source{S}$ to some inequivalent states.

\begin{theorem}
    If $ \gamma $ respects behaviors and $ \target{q}^\target{T} _w   \in   \target{Q}_{ \textsf{Weird} } $ is a weird state,
   then $ \target{q}^\target{T} _w   \in   \textsf{Exploit}   \ottsym{(}   \IFSM   \ottsym{)}$.
\end{theorem}

Importantly, the definition of exploits used in this paper includes weird machines that are not included in Dullien's version. Suppose there is a transition $ \target{q}^\target{T}_{{\mathrm{1}}}  \overstep{ s }  \target{q}^\target{T}_{{\mathrm{2}}} $ between two sane states in $ \CPU $ (meaning that there exist IFSM states $\source{q}^\source{S}_{\ottmv{i}}$ such that $\target{q}^\target{T}_{\ottmv{i}}  \in   \gamma   \ottsym{(}  \source{q}^\source{S}_{\ottmv{i}}  \ottsym{)}$) but there is no transition $ \source{q}^\source{S}_{{\mathrm{1}}}  \overstep{ s }  \source{q}^\source{S}_{{\mathrm{2}}} $ in the IFSM. This would not be considered an exploit in Dullien's definition, because there are no weird states involved, but the implementation clearly behaves in a way not intended by the IFSM.

\section{Towards understanding exploits}
\label{sec:discussion}
Ultimately, the reason why we are studying exploits is to make software systems more secure. In this section we sketch ways in which future work could make inroads towards understanding exploits and weird machines at a deeper level.

\subsection{What vulnerabilities exist?}

First, we can use this framework to ask: given a particular attacker class or set of language features, what kinds of vulnerabilities exist? As a concrete example, consider the first stage of the compiler from \cref{sec:compositionality}, which compiles a simple imperative language $ \IMP $ to a C-like language $ \ToyC $. The $ \IMP $ language has natural number and boolean values, arithmetic operations, if and while statements, and an $ \OUTPUT{ \ottmv{x} } $ statement that emits a value to an output trace. The $ \ToyC $ language extends these features with procedure calls, integer and pointer values, and pointer arithmetic.

The fact that natural numbers and booleans in $ \IMP $ are both represented as integers in $ \ToyC $ means that attackers can violate the data structure abstractions in $ \IMP $. Even a simple attacker class $ \mathsource{\Attack} $ in $ \ToyC $ with  no control flow or output capabilities can result in data-oriented programming vulnerabilities~\citep{Hu2015,Hu2016}.

In the compiler, components in $ \IMP $ are compiled to a procedure declaration $ \texttt{hole}   \ottsym{(}   \vect{ \ottmv{x_{\ottmv{i}}} }   \ottsym{;}  .\xspace  \ottsym{)}  \ottsym{\{}  \source{c}  \ottsym{\}}$ in $ \ToyC $. An attacker context $\source{A}  \in   \mathsource{\Attack} $ is a procedure called $ \texttt{main} $ that can invoke the component as follows:
\begin{align*}
  \source{A} &::=  \texttt{main}   \ottsym{(}   \vect{ \ottmv{x_{\ottmv{i}}} }   \ottsym{;}   \vect{ \ottmv{y_{\ottmv{j}}} }   \ottsym{)}  \ottsym{\{}  \source{a}  \ottsym{\}}  \tag{$ \ToyC $ DOP attacks} \\
  \source{a} &::= \ottmv{x}  \coloneqq  \source{e}  \ottsym{;}  \source{a} ∣  \texttt{hole}   \ottsym{(}  \ottmv{x_{{\mathrm{1}}}}  \ottsym{,} \, .. \, \ottsym{,}  \ottmv{x_{\ottmv{n}}}  \ottsym{)}
\end{align*}
That is, the attacker context can only initialize arguments provided to a compiled $ \IMP $ program.

Even though the set of attacker contexts $ \mathsource{\Attack} $ does not give access to the output trace, an attacker can use a weird machine to output negative numbers using the vulnerable source-level component $ \OUTPUT{ \ottmv{x} } $. In particular, for any $ \mathsource{ i }  \, \ottsym{<} \,  \mathsource{ \ottsym{0} } $, the context $ \ottmv{x}  ≔   \mathsource{ i }    \ottsym{;}   \texttt{hole}   \ottsym{(}  \ottmv{x}  \ottsym{)}$ is an exploit of $ \OUTPUT{ \ottmv{x} } $, and it is possible to classify the entire set of weird machines of $ \OUTPUT{ \ottmv{x} } $:
\[   \textsf{WM}  ^{  \mathsource{\Attack}  }   \ottsym{(}   \OUTPUT{ \ottmv{x} }   \ottsym{)} = \{  \mathsource{ i }  ∣  \mathsource{ i }  \, \ottsym{<} \,  \mathsource{ \ottsym{0} }  \} \]

This property can be extended to all weird machines between $\IMP$ and $ \ToyC $.

\begin{theorem} \label{thm:WMNZ}
    If an attacker context $\source{A}  \in   \mathsource{\Attack} $ causes $ \compile[  \intermediate{U}^\intermediate{I}  ] $ to output negative numbers---if
    there is some trace $\intermediate{t}^\intermediate{T}  \in   \sbehav{  \source{A} [   \compile[  \intermediate{U}^\intermediate{I}  ]   ]  } $ containing a negative number---then
    $\source{A}  \in    \textsf{T}  \textsf{Exploit}   ^{  \mathsource{\Attack}  }   \ottsym{(}  \intermediate{U}^\intermediate{I}  \ottsym{)}$.
\end{theorem}
\begin{proof}
    By \cref{thm:TraceExploit}, it suffices to show that, for all $ \IMP $ whole programs $\intermediate{P}^\intermediate{I}$, if $\intermediate{t}^\intermediate{T}  \in   \ibehav{ \intermediate{P}^\intermediate{I} } $ and $\source{t}^\source{S}  \in   \mathsource{ \mathcal{B} }^\source{S} $ contains negative numbers, then $ \nbrelates{ \intermediate{t}^\intermediate{T} }{ \source{t}^\source{S} } $; this follows from the definition of $ \mapsto $.
\end{proof}

\subsection{How powerful is a vulnerability?}

Understanding the class of weird machines $  \textsf{WM}  ^{  \mathtarget{\Attack}  }   \ottsym{(}  \source{V}  \ottsym{)}$ associated with a vulnerability helps us understand that vulnerability's strength.
For example, a vulnerability with hyperproperty exploits may be less powerful but harder to detect than a vulnerability with trace exploits. 

As another example, \citet{Hu2016} show that data-oriented programming is Turing complete given a large enough set of entry points into the source component. Restated in our framework, given an attacker class $ \mathtarget{\Attack} $ and a vulnerable source program with enough entry points, the class $  \textsf{WM}  ^{  \mathtarget{\Attack}  }   \ottsym{(}  \source{V}  \ottsym{)}$ is Turing complete. 

However, Turing completeness is not the ultimate goal of attackers; their goal is to \emph{easily and efficiently} program their attacks. \citet{Bratus2017} have proposed that we should instead ask if a set of attacks is \emph{compositional}: can the class $  \textsf{WM}  ^{  \mathtarget{\Attack}  }   \ottsym{(}  \source{V}  \ottsym{)}$ be composed out of smaller, easier-to-understand classes $  \textsf{WM}  ^{ \mathtarget{\Attack}_{{\mathrm{1}}} }   \ottsym{(}  \source{V}  \ottsym{)},…,  \textsf{WM}  ^{ \mathtarget{\Attack}_{\ottmv{n}} }   \ottsym{(}  \source{V}  \ottsym{)}$? If a vulnerability is compositional in this way, it is especially susceptible to attacks, since attacks can then be programmed rather than hand crafted.

\subsection{What mitigations are effective?}

Finally, by characterizing exploits we make it possible reason about the effectiveness of different mitigation strategies. Returning to the exploits of $ \OUTPUT{ \ottmv{x} } $, a programmer may wish to insert dynamic checks to eliminate the possibility of exploits. For example, the programmer or compiler could replace occurrences of $ \OUTPUT{ \ottmv{x} } $ in an $\IMP$ program with $ \ifThenElse{ \ottmv{x} \, \ge \,  \mathintermediate{ \ottsym{0} }  }{  \OUTPUT{ \ottmv{x} }  }{  \SKIP  } $. These two programs are equivalent in $\IMP$, but their compiled versions are inequivalent in $ \ToyC $. However, this program still has a non-empty weird machine, since the compiled machine admits the behavior with an empty output trace but the $\IMP$ program always emits an output. 
\small
\begin{align*}
    &  \textsf{WM}  ^{  \mathtarget{\Attack}  }   \ottsym{(}   \ifThenElse{ \ottmv{x} \, \ge \,  \mathintermediate{ \ottsym{0} }  }{  \OUTPUT{ \ottmv{x} }  }{  \SKIP  }   \ottsym{)}
    = \{ε\}.
\end{align*}
\normalsize
If the $\ELSE$ branch of the dynamic check is
replaced with the statement $ \OUTPUT{ \ottsym{0} } $, then the resulting program will have no exploits, since the behavior of all target contexts will be simulatable by some $ \IMP $ context.
\small
\begin{align*}
    &  \textsf{WM}  ^{  \mathtarget{\Attack}  }   \ottsym{(}   \ifThenElse{ \ottmv{x} \, \ge \,  \mathintermediate{ \ottsym{0} }  }{  \OUTPUT{ \ottmv{x} }  }{  \OUTPUT{  \mathintermediate{ \ottsym{0} }  }  }   \ottsym{)} \\
    & = ∅.
\end{align*}
\normalsize

Although we are illustrating these mitigations with an extremely simple example, prior work has shown that secure compilation techniques can be used to reason about mitigations that protect against more realistic exploits, such as address space layout randomization~\citep{Abadi2012}.

\section{Related Work}
\label{sec:related}

\subsection{Weird machines}
The idea of a ``weird machine'' was introduced by Sergey Bratus to connect the act of exploitation with programming~\citep{Bratus2009} and further developed by the LangSec research community as a way to understand the causes of computer insecurity~\citep{Bratus2011}. In addition to using these ideas to describe existing exploitation techniques such as return-oriented programming~\citep{Shacham2007} and heap feng shui~\citep{Sotirov2007}, researchers have identified weird machines in a variety of systems including DWARF exception handling data~\citep{Oakley2011}, the ELF executable format~\citep{Shapiro2013}, and embedded system interrupt handlers~\citep{Tan2014}. A common approach to demonstrating the generality of weird machines and exploit techniques such as return-oriented programming~\citep{Shacham2007}, data-oriented programming~\citep{Hu2016}, and counterfeit object-oriented programming~\citep{Schuster2015} is to show that the computations they enable are Turing complete.

\citet{Bratus2017} argue the necessity of a theory of weird machines and suggest a model relating two abstraction levels of a program. Using the theory discussed in Section~\ref{sec:weird-behavior}, \citet{Dullien2018} uses weird machines to reason about the exploitability of different program implementations. His work builds on a transducer model of weird machines introduced by Vanegue~\citep{Vanegue2014}.

\subsection{Secure compilation and exploits}

\citet{Abadi1999} introduced full abstraction as a security criterion for compilation. Researchers have studied the development of secure compilers for a variety of source and target languages, and an excellent survey is given by \citet{Patrignani2018}.

\citet{Kennedy2006} uses full abstraction to examine security problems in the {.NET} programming model and propose fixes, and \citet{Abadi2012} use full abstraction to prove the effectiveness of address space layout randomization as a security mitigation. \citet{Erlingsson2007} suggests many vulnerabilities stem from the lack of secure compilation and discusses how mitigation strategies seek to maintain certain properties of high-level languages.

\citet{Patrignani2017} describe several shortcomings of full abstraction as a definition of secure compilation and introduce property preservation as an alternative definition. \citet{Abate2019} build on this work to explore a hierarchy of robust property preservation definitions, which \citet{Abate2019a} later extend to support languages whose source and target behaviors differ. We describe the relationship between our framework and this hierarchy in Section~\ref{sec:relprop}.

\section{Conclusion}
\label{sec:conclusion}



This paper presents a framework for reasoning about weird machines as insecure compilation. We define exploits  as attacks that behave in ways inaccessible in the source language, and show how examples of common exploits fit into our framework. We draw on prior work about secure compilation to prove that exploits always violate properties of source language behaviors, and distinguish between different classes of properties, such as hyperproperties and trace properties. We show how exploits compose through a compiler stack, and compare our approach to prior work by \citet{Dullien2018} that characterizes exploits by the states they invoke rather than the behavior they produce.

We draw heavily on programming language techniques rather than state machine formalisms, making it possible to reason directly about the behavior of programs and compilation and eliminating the need to manually construct a model of the relationship between an implementation and its intended behavior.

Though we use of the theory of secure compilation, our focus is quite different from that in the literature.  Most existing work attempts to construct secure compilers; we take for granted that industrial compilers fall short of that and study the repercussions.  This may appear cynical, but ultimately we hope that studying exploits leads to a future with fewer and fewer of them.

\bibliography{bibliography}
\bibliographystyle{IEEEtranN} 

\appendix[Examples]


\label{app:examples}

In this section we describe a compiler stack for a high-level imperative language $ \IMP $ with natural numbers and boolean values. This language is first compiled to a C-like language $ \ToyC $ with procedure calls and pointer arithmetic, then to a toy assembly language $ \ToyA $ with an explicitly managed stack. $ \ToyA $ is implemented in a model of hardware, $ \ToyH $, in which users can observe the timing information of $ \ToyA $ programs. Each stage in this compilation chain exposes different weird machines. 

The languages and exploits described here are also implemented in PLT Redex, made available as ancillary files on arXiv. 


\subsection{$ \IMP $}
\label{sec:IMP}

The first language on the compiler stack is $ \IMP $, a simple imperative language with natural numbers and boolean expressions. An $ \IMP $ program is a command---a sequence of assignments, control flow, and outputs.
\begin{align*}
   \intermediate{e} &::= \ottmv{x} ∣  \mathintermediate{ b }   \in   \textsf{Bool}  ∣  \mathintermediate{ n }   \in   \textsf{Nat}  
        ∣ \intermediate{e} \, \circledast \, \intermediate{e}
        \tag{$ \IMP $ expressions} \\
   \circledast &::= + ∣ × ∣ < ∣ == ∣ ∧ ∣ ∨ \tag{binary operations} \\
   \intermediate{c} &::= \ottmv{x}  \coloneqq  \intermediate{e} ∣  \OUTPUT{ \intermediate{e} }  ∣  \SKIP  ∣ \intermediate{c}  \ottsym{;}  \intermediate{c} \\
            &∣  \ifThenElse{ \intermediate{e} }{ \intermediate{c} }{ \intermediate{c} }  ∣  \while{ \intermediate{e} }{ \intermediate{c} } 
            \tag{$ \IMP $ commands}
\end{align*}
The operational semantics of $ \IMP $ is given as a judgment $ \intermediate{c}  \ottsym{/}  \mathintermediate{\sigma}  \overstep{ t }  \intermediate{c}'  \ottsym{/}  \mathintermediate{\sigma}' $, where: 
\begin{align*}
    \mathintermediate{\sigma} &::= \ottmv{x}  \mapsto  v ∣ \mathintermediate{\sigma}_{{\mathrm{1}}}  \ottsym{,} \, .. \, \ottsym{,}  \mathintermediate{\sigma}_{\ottmv{n}} \tag{$ \IMP $ store} \\
    t &::=  \cdot  ∣ v  \ottsym{,}  t \tag{$ \IMP $ trace} \\
    v &::=  \mathintermediate{ b }   \in   \textsf{Bool}  ∣  \mathintermediate{ n }   \in   \textsf{Nat}  \tag{$ \IMP $ values}
\end{align*}


We write $ \intermediate{c}  \ottsym{/}  \mathintermediate{\sigma}  \not\rightarrow $ when $\intermediate{c}/\mathintermediate{\sigma}$ is a stuck configuration that cannot take a step, and we write $→^∗$ for the transitive closure of $→$. 
Finally, we write $\intermediate{c}  \ottsym{/}  \mathintermediate{\sigma}  \Downarrow  t$ when there exists some stuck state
$\intermediate{c}'  \ottsym{/}  \mathintermediate{\sigma}'$ such that $ \intermediate{c}  \ottsym{/}  \mathintermediate{\sigma}  \oversteps{ t }  \intermediate{c}'  \ottsym{/}  \mathintermediate{\sigma}' $.

\paragraph*{Weird machine meta-variables}

A whole program in $ \IMP $ is a command $\intermediate{c}$, and
the behavior of such a program is the set containing the output traces obtained by instantiating the free variables of $\intermediate{c}$. A component in $ \IMP $ is a command annotated by (a superset of) the free variables that occur in it. A context is a command with a hole in it; linking a context with a component fills the hole with the component, and is only defined when the annotation of the component matches the annotation of the hole. The annotations will be important for compilation.
\begin{align*}
    \intermediate{P}^\intermediate{I} &::= \intermediate{c} \tag{$ \IMP $ whole program} \\
    \intermediate{U}^\intermediate{I} &::= \ottsym{(}  \intermediate{c}  \ottsym{,}   \vect{ \ottmv{x_{\ottmv{i}}} }   \ottsym{)} \tag{$ \IMP $ component} \\
    \intermediate{C}^\intermediate{I} &::=  \square^{  \vect{ \ottmv{x_{\ottmv{i}}} }  }  ∣ \intermediate{C}^\intermediate{I}  \ottsym{;}  \intermediate{c} ∣ \intermediate{c}  \ottsym{;}  \intermediate{C}^\intermediate{I} \\
        &∣  \ifThenElse{ \ottnt{e} }{ \intermediate{C}^\intermediate{I} }{ \intermediate{c} }  ∣  \ifThenElse{ \ottnt{e} }{ \intermediate{c} }{ \intermediate{C}^\intermediate{I} }   \\
        &∣  \while{ \ottnt{e} }{ \intermediate{C}^\intermediate{I} } 
        \tag{$ \IMP $ context}
\end{align*}


\subsection{$ \ToyC $}
\label{sec:ToyC}

We implement $ \IMP $ by compiling it to a C-like language called $ \ToyC $, where natural numbers and booleans are both implemented as integers. 
Expressions in $ \ToyC $ consist of  integers, unary and binary operations, and pointers. Pointers are represented as l-values, $ℓ^{v}$, which are either variables or pointer dereferences.
\small
\begin{align*}
    \source{e} &::= \ottmv{x} ∣  \mathsource{ i }  ∣ \source{e}_{{\mathrm{1}}} \, \circledast \, \source{e}_{{\mathrm{2}}} ∣ \ottkw{NULL} ∣ \ottsym{*}  \source{e} ∣ \ottsym{\&}  ℓ^{v}  \tag{$ \ToyC $ expressions} \\
    ℓ^{v} &::= \ottmv{x} ∣ \ottsym{*}  ℓ^{v} \tag{$ \ToyC $ l-values}
\end{align*}
\normalsize
Expressions can be assigned types $\mathsource{\tau}$--either simple types (integers and pointers) or arrays of simple types. 
\begin{align*}
  \mathsource{\tau}  &::= \mathsource{\sigma} ∣ \mathsource{\sigma}  \ottsym{[}  n  \ottsym{]} \tag{$ \ToyC $ types}\\
  \mathsource{\sigma} &::=  \textsf{Int}  ∣ \ottsym{*}  \mathsource{\tau} \tag{$ \ToyC $ simple types}
\end{align*}


Commands in $ \ToyC $ are similar to those in $ \IMP $ with the addition of procedure calls $p  \ottsym{(}  \source{e}_{{\mathrm{1}}}  \ottsym{,} \, .. \, \ottsym{,}  \source{e}_{\ottmv{n}}  \ottsym{)}$. Assignment statements are extended from variables to l-values.
\small
\begin{align*}
    &\source{c} ::= p  \ottsym{(}  \source{e}_{{\mathrm{1}}}  \ottsym{,} \, .. \, \ottsym{,}  \source{e}_{\ottmv{n}}  \ottsym{)} ∣  ℓ^{v}  ≔  \source{e}  ∣  \OUTPUT{ \source{e} } 
                ∣  \SKIP  ∣ \source{c}  \ottsym{;}  \source{c} \\
                ∣ & \ifThenElse{ \source{e} }{ \source{c} }{ \source{c} }  ∣  \while{ \source{e} }{ \source{c} } 
    \tag{$ \ToyC $ commands}
\end{align*}
\normalsize
A global store $\mathsource{G}$ is a list of procedure declarations, which map procedure names $p$ to commands $\source{c}$ while specifying the types of their arguments as well as any local variables used in the procedure. Note that arguments to procedures must all have simple types.
\small
\begin{align*}
    \source{pd} &::= p  \ottsym{(}   \vect{ \ottmv{x_{\ottmv{i}}}  \ottsym{:}  \mathsource{\sigma}_{\ottmv{i}} }   \ottsym{;}   \vect{ \ottmv{y_{\ottmv{j}}}  \ottsym{:}  \mathsource{\tau}_{\ottmv{j}} }   \ottsym{)}  \ottsym{\{}  \source{c}  \ottsym{\}} 
        \tag{$ \ToyC $ procedure declaration} \\
    \mathsource{G} &::= \ottsym{\{}  \source{pd}_{{\mathrm{1}}}  \ottsym{,} \, .. \, \ottsym{,}  \source{pd}_{\ottmv{n}}  \ottsym{\}}
        \tag{$ \ToyC $ global store}
\end{align*}
\normalsize
%
%
The language is ``C-like'' in that it respects the call structure of C and has pointer arithmetic, but of course it differs from C in many ways. For one, $ \ToyC $ will throw an error when it reaches undefined behavior, unlike C, which may execute any behavior (see \cref{sec:undefined}). 

\paragraph*{Weird machine meta-variables}

A whole program in $\ToyC$ is a global store $\mathsource{G}$ such that (1) there is a function called $ \texttt{main} $, (2) no two functions have the same name, and (3) every procedure name $p  \ottsym{(}  \ottnt{e_{{\mathrm{1}}}}  \ottsym{,} \, .. \, \ottsym{,}  \ottnt{e_{\ottmv{n}}}  \ottsym{)}$ occurring in $\mathsource{G}$ has a corresponding procedure declaration. Components $\source{U}^\source{S}$ and contexts $\source{C}^\source{S}$ correspond to stores $\mathsource{G}$ satisfying condition (2), but not necessarily (1) or (3). Linking a component with a context simply concatenates the two stores. 
\begin{align*}
  \source{P}^\source{S} &::= \mathsource{G} ~\text{such that}~ \texttt{main}  ∈ \dom*{\mathsource{G}} 
        \tag{$ \ToyC $ whole program} \\
  \source{U}^\source{S} &::= \mathsource{G} \tag{$ \ToyC $ component} \\
  \source{C}^\source{S} &::= \mathsource{G} \tag{$ \ToyC $ context} 
\end{align*}

The behavior of a whole program is the set of traces that invocations of $ \texttt{main} $ can give rise to. 

\subsection{Compiler from $ \IMP $ to $ \ToyC $.}

The compiler from $ \IMP $ to $ \ToyC $ maps $ \IMP $ commands $\intermediate{c}$ to $ \ToyC $ commands $ \ccompile{ \intermediate{c} } $, and then to procedure declarations $ \texttt{main}   \ottsym{(}   \vect{ \ottmv{x_{\ottmv{i}}} }   \ottsym{;}  .\xspace  \ottsym{)}  \ottsym{\{}   \ccompile{ \intermediate{c} }   \ottsym{\}}$ where $ \vect{ \ottmv{x_{\ottmv{i}}} } $ are the free variables used in $\intermediate{c}$. The full compiler is implemented in the PLT Redex model.

The behavior relation relates boolean/natural number traces in $ \IMP $ to their corresponding integer traces in $ \ToyC $; we define $ \compile[  t  ] $ so that $\brelates{t}{ \compile[  t  ] }$:
\[ \begin{aligned}
     \compile[  .\xspace  ]  &≜  \cdot  \\
     \compile[   \mathintermediate{ \true }   \ottsym{,}  t  ]  &≜  \mathsource{ \ottsym{1} }   \ottsym{,}   \compile[  t  ] 
\end{aligned} \qquad \begin{aligned}
     \compile[   \mathintermediate{ \false }   \ottsym{,}  t  ]  &≜  \mathsource{ \ottsym{0} }   \ottsym{,}   \compile[  t  ]  \\
     \compile[   \mathintermediate{ n }   \ottsym{,}  t  ]  &≜  \mathsource{ n }   \ottsym{,}   \compile[  t  ] 
\end{aligned} \]

\begin{lemma} 
  The compiler from $ \IMP $ to $ \ToyC $ is not correct for whole programs.
\end{lemma}
\begin{proof}
Consider the whole program $\ottmv{x}  \coloneqq  \ottsym{1}  \ottsym{;}   \ifThenElse{ \ottmv{x} }{  \OUTPUT{  \mathintermediate{ \true }  }  }{  \OUTPUT{  \mathintermediate{ \false }  }  } $. This is ill-typed in $ \IMP $, so its output trace is empty. But when compiled to $ \ToyA $ there is no such mismatch between booleans and numbers, so it will execute successfully and output $ \compile[   \mathintermediate{ \true }   ] = \mathtarget{ \ottsym{1} } $.
\end{proof}

However, the compiler does preserve the traces of whole programs (\cref{sec:compositional-trace}):
\begin{lemma}
  If $t  \in   \ibehav{ \intermediate{P}^\intermediate{I} } $ then $ \compile[  t  ]   \in   \tbehav{  \compile[  \intermediate{P}^\intermediate{I}  ]  } $.
\end{lemma}
\begin{proof}
  If $t  \in   \ibehav{ \intermediate{P}^\intermediate{I} } $ then there is some store $\mathintermediate{\sigma}$ such that $ \intermediate{P}^\intermediate{I}   \ottsym{/}  \mathintermediate{\sigma}  \Downarrow  t$. But then the values of that store can be converted to inputs to the compiled $ \texttt{main} $ procedure of $ \compile[  \intermediate{P}^\intermediate{I}  ] $, which will produce the same trace in $ \ToyC $.
\end{proof}

\subsection{Data-oriented programming with $ \IMP $ and $ \ToyC $}
\label{sec:dop}

It is impossible to write an $ \IMP $ program that performs natural number division, which requires subtraction. However, it is possible to use data-oriented programming on a benign-looking $ \IMP $ component to implement this behavior. Consider the following $ \IMP $ component:
\[\begin{aligned}
    \intermediate{J} ≜~ 
    &\OUTPUT{x} \\
    &\OUTPUT{y} \\
    &r ≔ y; \\
    &d ≔ 0; \\
    &\while{x≤r}{\{}&r &≔ r+a \\
    &    &x &≔ x+b \\
    &    &d &≔ d+1 \}; \\
    &\OUTPUT{d}; \\
    &\OUTPUT{r};
\end{aligned}
\]
In $ \IMP $, $\intermediate{J}$ will terminate either if $x$ is initially greater than $r$, or if $b$, the amount by which $x$ is increasing, is more than $a$, the amount by which $r$ is increasing. The counter $d$ tracks the number of iterations the loop takes. The behavior of $\intermediate{J}$ in $ \IMP $ will always output the trace $x,y,d,r$ such that $r ≥ y$; this is the trace property by which we will measure weirdness.

An attacker in $ \ToyC $ can manipulate the arguments to $\intermediate{J}$ to implement division (which will output $x,y,d,r$ such that $y=dx+r$ and $r<x$) as follows:
\begin{align*}
    \source{A} ≜  \texttt{main} (x,y;a,b,d,r)
            \{ &a ≔ -x; b ≔ 0; d ≔ 0; r ≔ 0; \\
               & \texttt{hole} (x,y,d,r,a,b) \}
\end{align*}

To check that $\source{A}  \in   \textsf{WM}   \ottsym{(}  \intermediate{J}  \ottsym{)}$, it suffices to observe some trace in $ \sbehav{  \source{A} [   \compile[  \intermediate{J}  ]   ]  } $ that outputs $x,y,d,r$ such that $r < y$; this is obtained with $ \texttt{main}   \ottsym{(}   \mathsource{ \ottsym{3} }   \ottsym{,}   \mathsource{ \ottsym{4} }   \ottsym{)}$, which outputs $3,4,1,2$.

This example is implemented in PLT Redex (the program \texttt{dop-vulnerable} in \texttt{toyc.rkt}).

\subsection{Toy Assembly}

In the toy assembly language $\ToyA$, statements as well as objects are stored in memory. Control flow is limited to procedure calls and conditional jumps, and procedures must explicitly invoke $ \RETURN $ to return to their calling function. The command $ \HALT $ halts the computation.
\small
\begin{align*}
    \target{c} &::= p  \ottsym{(}  \target{e}_{{\mathrm{1}}}  \ottsym{,} \, .. \, \ottsym{,}  \target{e}_{\ottmv{n}}  \ottsym{)} ∣  \RETURN  ∣  \HALT  ∣  ℓ^{v}  ≔  \target{e}  ∣  \OUTPUT{ \target{e} }  \\
    &∣  \SKIP  ∣  \jmpz{ \target{e} }{ i } 
    \tag{$ \ToyA $ commands}
\end{align*}
\normalsize
Expressions in $ \ToyA $ are identical to those in $ \ToyC $.

Memory in $\ToyA$ is a map from integers to objects, where objects are either values or statements of the form $p  \rhd  \target{c}$ where $p$ is a procedure name.
\begin{align*}
    \mathtarget{M} &::= \ottsym{∅} ∣ \mathtarget{M}  \ottsym{,}  i  \mapsto  \target{o}^\target{T} 
        \tag{$ \ToyA $ memory} \\
    \target{o}^\target{T} &::= \mathtarget{v} ∣ p  \rhd  \target{c}
        \tag{$ \ToyA $ objects} \\
    \mathtarget{v} &::= i
        \tag{$ \ToyA $ values} \\
\end{align*}
Control flow in $\ToyA$ is managed by the program counter $\target{PC}$, a regular pointer into memory, and the stack is managed by the stack pointer $\target{SP}$.

Each procedure is associated with a frame, which maps variables to offsets at which the variables occur with respect to the stack pointer. Frames distinguish variables from arrays, which are also associated with their length. Procedure declarations map procedure names to their entry point as a program counter, as well as their associated stack frame. A global store is just a list of procedure declarations.
\small
\begin{align*}
    \mathtarget{F} &::= \,  \ottsym{[}  \,  \ottsym{]} ∣  \texttt{var}~ \ottmv{x} @ n  ∣  \texttt{array}~ \ottmv{x} @ n_{{\mathrm{1}}} ( n_{{\mathrm{2}}} ) 
        \tag{$ \ToyA $ frame} \\
    \mathtarget{pd} &::= p  \ottsym{(}  \target{PC}  \ottsym{)\{}  \mathtarget{F}  \ottsym{\}}
        \tag{$ \ToyA $ procedure declaration} \\
    \mathtarget{G} &::= \ottsym{\{}  \mathtarget{pd}_{{\mathrm{1}}}  \ottsym{,} \, .. \, \ottsym{,}  \mathtarget{pd}_{\ottmv{n}}  \ottsym{\}}
        \tag{$ \ToyA $ global store}
\end{align*}
\normalsize

\paragraph*{Weird machine meta-variables}
A whole program in $\ToyA$ is a configuration $\ottsym{<}  \target{PC}  \ottsym{;}  \target{SP}  \ottsym{;}  \mathtarget{M}  \ottsym{>}$ along with a global store $\mathtarget{G}$.
A component in $ \ToyA $ is a global store along with a memory fragment storing instructions. A context consists of a program counter, stack pointer, and another memory fragment that can contain both initial memory and further instructions. The two disjoint memory fragments are merged during linking.
\small
\begin{align*}
  \target{P}^\target{T} &::= \mathtarget{G}  \vdash  \ottsym{<}  \target{PC}  \ottsym{;}  \target{SP}  \ottsym{;}  \mathtarget{M}  \ottsym{>} \tag{$ \ToyA $ whole program}\\
  \target{U}^\target{T} &::= \mathtarget{G}  \vdash  \mathtarget{M} \tag{$ \ToyA $ component} \\
  \target{C}^\target{T} &::= \ottsym{<}  \target{PC}  \ottsym{;}  \target{SP}  \ottsym{;}  \mathtarget{M}  \ottsym{>} \tag{$ \ToyA $ context} \\
  &\ottsym{<}  \target{PC}  \ottsym{;}  \target{SP}  \ottsym{;}  \mathtarget{M}_{{\mathrm{0}}}  \ottsym{>}  \ottsym{[}  \mathtarget{G}  \vdash  \mathtarget{M}_{{\mathrm{1}}}  \ottsym{]} ≜ \mathtarget{G}  \vdash  \ottsym{<}  \target{PC}  \ottsym{;}  \target{SP}  \ottsym{;}  \mathtarget{M}_{{\mathrm{0}}}  \ottsym{,}  \mathtarget{M}_{{\mathrm{1}}}  \ottsym{>} \tag{$ \ToyA $ linking}
\end{align*}
\normalsize
The behavior of a whole program is the set of output traces it can produce, ending in a $ \HALT $ instruction. 

\subsection{Compiler from $ \ToyC $ to $ \ToyA $}
\label{sec:compiler-ToyC-ToyA}

The compiler from $ \ToyC $ to $ \ToyA $ maps a $ \ToyC $ command $\source{c}$ to a memory containing instruction pointers, which is merged with the $ \ToyC $ memory, which is layed out according to the $ \ToyC $ stack.

Like the $ \IMP $ compiler, the $ \ToyC $ compiler preserves the traces of whole programs:
\begin{theorem}
    Let $\source{P}^\source{S}$ be a whole program in $ \ToyC $. If $\source{t}^\source{S}  \in   \sbehav{ \source{P}^\source{S} } $ then there exists some $\target{t}^\target{T}$ such that $ \brelates{ \source{t}^\source{S} }{ \target{t}^\target{T} } $ and $\target{t}^\target{T}  \in   \tbehav{  \compile[  \source{P}^\source{S}  ]  } $.
\end{theorem}

The compiler is not correct for all whole programs because there are $ \ToyC $ programs that are undefined in $ \ToyC $ (such as those containing buffer overflows) but whose compiled versions do successfully execute. 

\begin{listing}
\begin{ccode}
void vulnerable (int* args, int len) {
  int[2] p;
  while (len > 0) {
    *p = *args;
    p += 1;
    args += 1;
    len -= 1;
  }
}
void store (int* p) {
    p = 42;
}
void print (int* p) {
    output *p;
}
\end{ccode}
\caption{A program vulnerable to return-oriented programming}
\label{lst:rop}
\end{listing}

\begin{listing}
\begin{ccode}
vulnerable(
  { sp0-15 /* desired stack pointer
            * on return from 'store' */
  , pc0 + 17 /* address of 'output *p' 
              * within 'print' */
  , sp0-6; /* target value of 'p' on iteration
            * that writes to 'p' */
  , 3; /* target value of 'len' on iteration 
        * that writes to 'len' */
  , sp0-15; /* desired stack pointer on return
             * from 'vulnerable' */
  , pc0+15; /* address of '*p = 42'
             * within 'store' */
  }, 6 ) /* value of len */
\end{ccode}
\caption{Attack that causes \texttt{vulnerable} to loop, printing 42 on each iteration. Here, \C{sp0} refers to the initial address of the stack pointer, and \C{pc0} refers to the initial address of the program counter}
\label{lst:ROP-attack}
\end{listing}

\subsection{Return-oriented programming}

In this section, we consider return-oriented programming attacks where, by carefully crafting the input to a program, we can manipulate the program counter and redirect it to malicious code. The attacker can provide arbitrary inputs to the $ \ToyC $ component, but cannot execute its own code.

Consider the $ \ToyC $ component in \cref{lst:rop} consisting of the procedures \texttt{vulnerable}, \texttt{store}, and \texttt{print}.
In $ \ToyC $, \texttt{vulnerable} is only defined when \texttt{len} is 0 or 1, since larger lengths would cause \texttt{p} to overflow.


The attack in \cref{lst:ROP-attack} passes an array of length 6 to \texttt{vulnerable} that
it to loop indefinitely, outputting $42$ each time. Since there is no context that results in an infinite loop according to the source semantics, this is an exploit. The example is implemented in the PLT Redex model (the program \texttt{rop-modular} in \texttt{toya.rkt}), and
we briefly sketch its execution here.

The attack causes 6 iterations of the loop in \texttt{vulnerable} to be executed. By the end of these six iterations, the function \texttt{vulnerable} will return to the program counter recorded at index $\target{SP}_{{\mathrm{0}}} - 9$, which will have been overwritten by the loop to point to the instruction $*p ≔ 42$ in \texttt{store}. Next, the computation will return from the function \texttt{store} to the program counter recorded at index $\target{SP}_{{\mathrm{0}}}-13$, which will have been overwritten to point to the instruction $\OUTPUT*{*p}$ in \texttt{print}; this will output the constant $42$. Finally, as the computation returns from the function \texttt{print}, the program counter will continue to be directed to the instruction $\OUTPUT*{*p}$, so the program will continue to loop forever, printing out $42$ at each step.

\subsection{$ \ToyH $}
\label{sec:toyh}

As the final layer in our compiler stack, we consider a hypothetical physical machine ($ \ToyH $ for toy hardware) on which toy assembly code is implemented. Programs in this machine are sequences of invocations of assembly programs. The language is similar to $ \IMP $ with booleans and integers, with the addition of a command $ ( \ottmv{x_{{\mathrm{1}}}} , \ottmv{x_{{\mathrm{2}}}} ) ≔ \texttt{get\_info}( \target{P}^\target{T} ) $ that executes the $ \ToyA $ program $\target{P}^\target{T}$ and stores both the output trace of the program and the number of steps it took to execute in the value $x$. That is, a programmer of $ \ToyA $ can execute $ \ToyA $ programs and use the timing information to make decisions.
\small
\begin{align*}
    \hardware{e} &::= \ottmv{x} ∣  \hardware{ b }   \in   \textsf{Bool}  ∣  \hardware{ i }   \in   \textsf{Int} 
                 ∣ \hardware{e}_{{\mathrm{1}}} \, \circledast \, \hardware{e}_{{\mathrm{2}}}
                 \tag{$ \ToyH $ expressions} \\
    \hardware{c} &::= \ottmv{x}  \coloneqq  \hardware{e}  ∣  ( \ottmv{x_{{\mathrm{1}}}} , \ottmv{x_{{\mathrm{2}}}} ) ≔ \texttt{get\_info}( \target{P}^\target{T} ) 
                ∣  \OUTPUT{ \hardware{e} }   \\
             &  ∣  \SKIP  ∣ \hardware{c}_{{\mathrm{1}}}  \ottsym{;}  \hardware{c}_{{\mathrm{2}}}
                ∣  \ifThenElse{ \hardware{e} }{ \hardware{c}_{{\mathrm{1}}} }{ \hardware{c}_{{\mathrm{0}}} }  \\
             &  ∣  \while{ \hardware{e} }{ \hardware{c} } 
                \tag{$ \ToyH $ commands}
\end{align*}
\normalsize
A component in $ \ToyH $ is just a component in $ \ToyA $, and an attacking context is a $ \ToyH $ program with a hole for such $ \ToyA $ components.
\small
\begin{align*}
    \hardware{P}^\hardware{H} &::= \hardware{c} \tag{$ \ToyH $ whole program} \\
    \hardware{U}^\hardware{H} &::= \target{U}^\target{T} \tag{$ \ToyH $ components} \\
    \hardware{A}^\hardware{H} &::=  ( \ottmv{x_{{\mathrm{1}}}} , \ottmv{x_{{\mathrm{2}}}} ) ≔ \texttt{get\_info}( \target{C}^\target{T} [\square]) 
            ∣ \hardware{A}^\hardware{H}  \ottsym{;}  \hardware{c} ∣ \hardware{c}  \ottsym{;}  \hardware{A}^\hardware{H} \\
           &∣  \ifThenElse{ \hardware{e} }{ \hardware{A}^\hardware{H} }{ \hardware{c} } 
            ∣  \ifThenElse{ \hardware{e} }{ \hardware{c} }{ \hardware{A}^\hardware{H} }  \\
           &∣  \while{ \hardware{e} }{ \hardware{A}^\hardware{H} } 
            \tag{$ \ToyH $ attacks}
\end{align*}
\normalsize



\subsection{Compiler from $ \ToyA $ to $ \ToyH $}

The compiler from $ \ToyA $ to $ \ToyH $ simply calls \texttt{get\_info} on the $ \ToyA $ program:
\small
\begin{align*}
     \compile[  \target{P}^\target{T}  ]^{x,y}  &≜  ( \ottmv{x} , \ottmv{y} ) ≔ \texttt{get\_info}( \target{P}^\target{T} )  \\
     \compile[  \target{U}^\target{T}  ]  &≜ \target{U}^\target{T} \\
     \compile[  \target{C}^\target{T}  ]^{x,y}  &≜  ( \ottmv{x} , \ottmv{y} ) ≔ \texttt{get\_info}( \target{C}^\target{T} [\square]) 
\end{align*}
\normalsize

Since the compiled code simply executes the $ \ToyA $ source code, it is correct for whole programs.

\subsection{Side-channel timing attack with $ \ToyH $}

We implement the side-channel attack in \cref{example:side-channel} in the PLT Redex model (the program \C{findpass} in \C{toyh.rkt}), which illustrates how to discover a password in a linear number of calls to a password checker.

\end{document}